%% file: AS.tex
\documentclass[11pt]{article}
\usepackage[english]{babel} % English language/hyphenation
\usepackage{amsmath,amsfonts,amsthm,amssymb} % Math packages
\usepackage{amssymb}
\usepackage{array}
\usepackage{fullpage}
\usepackage{enumerate}
\usepackage{enumitem}
\usepackage{xstring}
\usepackage{ifthen}
\usepackage{tikz}
\usepackage{tikz-cd} 
\usetikzlibrary{calc}
\usetikzlibrary{positioning}
\usepackage[nodisplayskipstretch]{setspace} % line spacing
\usepackage[]{natbib} %bib
\setcitestyle{numbers}
\setcitestyle{square}
\usepackage[unicode]{hyperref} %hyperlinks
\usepackage{xcolor} %colors
\usepackage{bm} % bold fonts maths
\usepackage{mathrsfs} %script font
\usepackage{import}
\usepackage{pdfpages}
\usepackage{transparent}

\usepackage{geometry}
\definecolor{linkcolor}{RGB}{109,71,106}

\definecolor{lam1}{HTML}{127DFF}
\definecolor{lam2}{HTML}{2980b9}
\definecolor{lam3}{HTML}{D04133}
\definecolor{lam4}{HTML}{1F61B2}

\definecolor{plot1}{RGB}{52, 101, 28}
\definecolor{plot2}{HTML}{4BAD95}

\hypersetup{
colorlinks=true,
linkcolor=lam3,
urlcolor=plot2,
citecolor=lam2        % color of links to bibliography
}

\newtheorem{thm}{Theorem}
\newtheorem{prop}[thm]{Proposition}
\newtheorem{lem}{Lemma}

\theoremstyle{def}

\theoremstyle{remark}
\newtheorem{ex}{Example}

\makeatletter
\newcommand{\mylabel}[2]{#2.\def\@currentlabel{#2}\phantomsection\label{#1}}
\makeatother

\usepackage{titlesec}

\titleformat*{\subsection}{\bfseries\centering}
\titleformat*{\subsubsection}{\itshape}
\titleformat*{\paragraph}{\large\bfseries\centering}
\titleformat*{\subparagraph}{\large\bfseries\centering}

\titlespacing\section{0pt}{10pt plus 4pt minus 2pt}{4pt plus 0pt minus 2pt}
\titlespacing\subsection{0pt}{12pt plus 4pt minus 2pt}{2pt plus 0pt minus 2pt}
\titlespacing\subsubsection{0pt}{6pt plus 4pt minus 2pt}{2pt plus 0pt minus 2pt}

\renewenvironment{abstract}
 {\small
  \begin{center}
  \bfseries \abstractname\vspace{-.5em}\vspace{0pt}
  \end{center}
  \list{}{%
    \setlength{\leftmargin}{14mm}% <---------- CHANGE HERE
    \setlength{\rightmargin}{\leftmargin}%
  }%
  \item\relax}
 {\endlist}
 
 \renewcommand{\blacksquare}{\vrule height7pt width4pt depth1pt}

\makeatletter
\newcommand{\itm}[1]{%
\item[\textsc{#1.}]\protected@edef\@currentlabel{#1}%
}
\makeatother

\def \R{\mathbb{R}}

\def \W{\Omega}
\def \w{\omega}
\def \p{\mathscr{P}}
\def \P{\mathbb{P}}
\def \L{\mathcal{L}}

\def \F{\mathscr{F}}
\def \B{\textbf{B}}
\def \RB{\textup{\textbf{RB}}}
\def \2{\textup{\textbf{2W}}}

\def \P{\mathbb{P}}
\def \Q{\mathbb{Q}}

\renewcommand{\int}{\textup{int}}

\def \U{\mathscr{U}}

\def \>{\rangle}
\def\<{\langle}
\def \r{\textup{R}}

\def \rr{f}
\def \aa{f^A}

\newcommand{\ax}[1]{\textup{\textbf{#1}}}

\renewcommand{\phi}{\varphi}

\newenvironment{subproof}[2][\proofname]{%
  \renewcommand{\qedsymbol}{#2}%
  \begin{proof}[#1]%
}{%
  \end{proof}%
}

\begin{document}

\onehalfspacing

\title{\vspace{-6ex}\textsc{Algebraic Semantics for Relative Truth, Awareness, and Possibility 
%\thanks{I would like to thank .}
}}

\author{Evan Piermont\thanks{Royal Holloway, University of London, Department of Economics, \texttt{evan.piermont@rhul.ac.uk}}}

\maketitle

\begin{abstract}

This paper puts forth a class of algebraic structures, \emph{relativized Boolean algebras} (RBAs), that provide semantics for propositional logic in which truth/validity is only defined relative to a local domain. In particular, the join of an event and its complement need not be the top element. Nonetheless, behavior is locally governed by the laws of propositional logic. By further endowing these structures with operators---akin to the theory of modal Algebras---RBAs serve as models of modal logics in which truth is relative. In particular, modal RBAs provide semantics for various well known awareness logics and an alternative view of possibility semantics.
	
\bigskip
\noindent \emph{Key words: Relativized Boolean Algebras; Awareness Frames; Awareness Logics.}
\end{abstract}

\setlength{\abovedisplayskip}{3pt}
\setlength{\belowdisplayskip}{3pt}

\section{Introduction}

In many applications of logic, it is desirable that truth (and validity) is not defined globally, but relative to some local domain. Nonetheless, we may still want the logic to behave in a classical manor when examined locally, that is, when fixing the domain. For example:

\begin{itemize}
\item One may not want to discuss the properties of objects when they do not exist, or the necessity, knowledge, or obligation of statements when they are not defined. That is, different domains may represent differential states of existence. 
\item A special case of the above: one may want to consider agents who have differential awareness. Here, different domains represent the agents' different conceptions of what might exist.  In fact, this relative definition of truth has become commonplace in the epistemic formalization of (un)awareness, where an agent's reasoning is restricted by her awareness but is otherwise rational (hence classical) on her local domain of awareness \citep{heifetz2006interactive, li2009information}.
\item Another special case: when considering dynamic environments, one may want to discuss truth \emph{at a certain point in time}, relative to the current extant state of affairs. Here domains would be linearly ordered and indexed by time. 
\item One may want to consider truth in a system that is only partially determined, in other words, allowing some statements to be neither true nor false. Here, different domains represent different levels of resolution, and, requiring the determined statements to be logically consistent again requires locally classical behavior.
\item A special case of this: one may want to consider an agent who envisages, hence reasons about, vague worlds. Possibility semantics model partial resolution in this way, within a possible worlds framework \citep{humberstone1981worlds, benthem2016bimodal}.
\end{itemize}

This paper puts forth a class of algebraic structures, \emph{relativized Boolean algebras} (RBAs), that provide semantics for propositional logic in which truth is only defined relative to a local domain, but within a given domain behavior is classical. %Moreover, these domains are ordered and project downward in the order. 
By further endowing these structures with operators---akin to the theory of modal algebras or Boolean algebras with operators---RBAs serve as models of modal logics in which truth is relative. In particular, RBAs can serve as model of differential existence (showcased by Example \ref{ex:RBA}), unawareness and knowledge under unawareness (Example \ref{ex:fk}) and partial resolution (Example \ref{ex:poss}). 

 Like a Boolean algebra and RBA is a set endowed with meet, join, and negation operations, and bottom and top elements: $\RB = \<RB,\land,\lor,\neg,0,1\>$. These operations satisfy the axioms of Boolean algebras except $X \lor \neg X$, which we denote by $1_X$, need not be the top element, and 0 need not be the identity for $\lor$. The elements of RBAs can be ordered via the usual condition $Y \geq X$ iff $X \land Y = X$.

In place of these Boolean axioms are the weakened versions $X \lor 1 = 1_X$ and $X \land 0 = 0$. Theorem \ref{thm:relative} shows that under these two relaxations $\pi_1(X) = \{ Z \mid 1_Z = 1_X\}$ is itself a Boolean algebra. Hence, if we think of $\pi_1(X)$ as the \emph{domain} on which the truth of $X$ is defined, then within a domain, truth behaves classically.
 
Without additional structure, these domains bear little relation to one another.  Of course, in the applications reference above, the various domains are related: becoming more aware or resolving some vagueness generally does not overturn all previously held truths. To restrict how truth in one domain relates to truth in another, we add an additional requirement to the definition of RBAs.
The property---that $1_Y \geq 1_X$  implies $\neg(Y \land 1_X)= \neg Y \land 1_X$---ensures that if $Y \geq X$ then $Z \mapsto Z \land 1_X$ is a Boolean homomorphism from $\pi_1(Y)$ to $\pi_1(X)$. Hence, RBAs are naturally equipped with an ordering on domains and a sense of projection between them.

\begin{ex}
\label{ex:RBA}
There are two planets, one whose atmosphere filters out all but blue light (the blue planet) and the other whose atmosphere filters all but red light (the red planet). The sentient inhabitants of each planet have been working hard to classify the plant life they see. The red taxonomists have concluded that all plants use one of two methods of photosynthesis: $X_R$ or $Y_R$. The blue taxonomists have, in contrast, observed three types of photosynthesis, $X_B$, $Y_B$, and $Z_B$.

Although the planets are distant in time and space,
% and each planet's scientists are completely unaware of the other's existence, 
it turns out that red plants, should they somehow be transported to the blue planet, would continue to photosynthesize as usual---that is, the pairs $(X_R, X_B)$, and $(Y_R,Y_B)$ refer to the same method of photosynthesis. $Z_B$ plants, however, require blue light specifically, and such a process does not and cannot exist on the red planet. 

Red taxonomists have concluded that ``if not $X_R$ then $Y_R$". While from a universal perspective, this statement is false, it feels unsatisfactory to call the red scientists wrong: within their local domain of their existence $Y_R$ is indeed the negation of $X_R$. However, for rather obvious reasons,\footnote{I.e., if we wish to maintain also ``$X_R$ if and only if $X_B$'' and ``$Y_R$ if and only if $Y_B$'' then ``not $X_R$ implies $Y_R$" implies the impossibility of $Z_B$.} the machinery of classical logic does not permit the red scientist's conclusion. 

RBA's offer a way of reconciling conflicting local and universal observations, by defining truth relative to a domain. To wit:  
Let $RB$ consist of the union of the elements of Boolean Algebras, $\B$ (for blue) and $\textup{\textbf{R}}$ (for red), generated by $\{X_B,Y_B,Z_B\}$ and $\{X_R,Y_R\}$, respectively. Moreover, define the Boolean homomorphism $h^B_R: \B \to  \textup{\textbf{R}}$ defined by $X_B \mapsto X_R$, $Y_B \mapsto Y_R$ and $Z_B \mapsto 0_R$. The operations on $\RB$, when restricted to either Boolean algebra,  coincide with the Boolean operations thereon. For $W_B \in \B$ and $W_R \in \textup{\textbf{R}}$, set $W_B \land W_R = h^B_R(W_B) \land W_R$, and $W_B \lor W_R = h^B_R(W_B) \lor W_R$. The top element is $1_B$ and the bottom is $0_R$. This algebra is visualized by Figure \ref{fig:ex_RBA}.\footnote{Although every RBA can be constructed as the disjoint union of Boolean algebras and homomorphisms between them, as above, these Boolean algebras do not need to be ordered (in the sense that the homomorphisms are surjective) as in Example \ref{ex:RBA}.}

$\RB$ models the situation in which $\neg X_R$ is indeed $Y_R$, but  $X_R \lor Y_R = 1_R \neq 1$. That is, where ``if not $X_R$ then $Y_R$" is valid in the particular sense that it is true whenever it is defined, but it is not defined universally. Likewise since $X_B$ projects to $Y_B$, they refer to the same event \emph{should both be defined}.  
\qed
\end{ex}

\begin{figure}   
    \centering 
    \begin{minipage}{.9\textwidth}
        \centering
    \def\svgwidth{.8\columnwidth}
    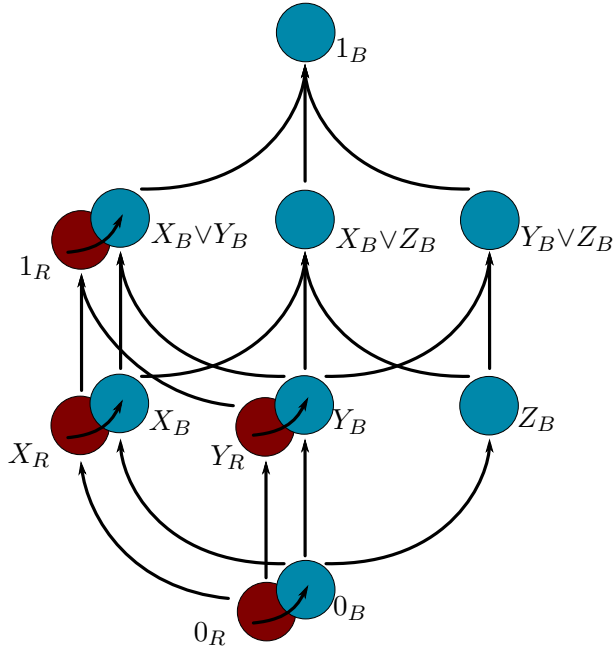

    \end{minipage}
    \caption{The RBA from Example \ref{ex:RBA}. The arrows indicate the partial ordering $\geq$. The blue elements compose $\B$, and the red elements, $\textup{\textbf{R}}$.}
    \label{fig:ex_RBA}
\end{figure}

Just as powersets serve as concrete examples of Boolean algebras,  given a set $W$, we can defined a \emph{concrete} RBA over  
$$\{(A,B) \mid B \subseteq W, A \subseteq B\}.$$
and with operations defined by

\begin{enumerate}[leftmargin=2cm]
\item[\mylabel{c}{\textsc{neg}}] $\neg(A,B) = ( B\setminus A, B)$
\item[\mylabel{c}{\textsc{meet}}] $(A,B) \land (A',B') = (A \cap A' , B \cap B')$
\item[\mylabel{c}{\textsc{join}}] $(A,B) \lor (A',B') = ((A \cup A') \cap (B \cap B'), B \cap B')$
\end{enumerate}

%Concrete RBAs make apparent the connection to HMS state-spaces.
%For each $B \subseteq W$, $\pi_1(B) = \{(A,B) \mid A \subseteq B\}$ is a local state-space and these state-spaces themselves are ordered via set inclusion. 
Theorem \ref{thm:embed} is a Stone-like representation theorem, showing that every RBA can be embedded into a concrete RBA. This inclusion, for the RBA considered in Example \ref{ex:RBA} is shown in Figure \ref{fig:ex_CRBA}.

As hinted at in Example \ref{ex:RBA}, RBAs serve as models of propositional logic in which truth and validity are relative by considering a homomorphism, $h: \L\to \RB$, between a propositional language, $\L$, and an RBA (i.e., a map such that $h(\neg\phi) = \neg h(\phi)$, $h(\phi \land \psi) = h(\phi) \land h(\psi)$, etc.) In concrete RBAs, the association between the formula $\phi$ and the event $(A,B)$ is intended to be thought of as specifying that $\phi$ is defined at $B$ and true at $A$. Hence the complement of $(A,B)$ is not $(A^c,B^c)$ but rather $(B\setminus A, B)$---the event where $\phi$ is defined but not true. A similar interpretation holds for the meets and joins. $\phi$ is valid in $\RB$ if $h(\phi) = (B,B)$; if $\phi$ is true wherever it is defined.

\begin{figure}[]
    \centering
    \def\svgwidth{.8\columnwidth}
    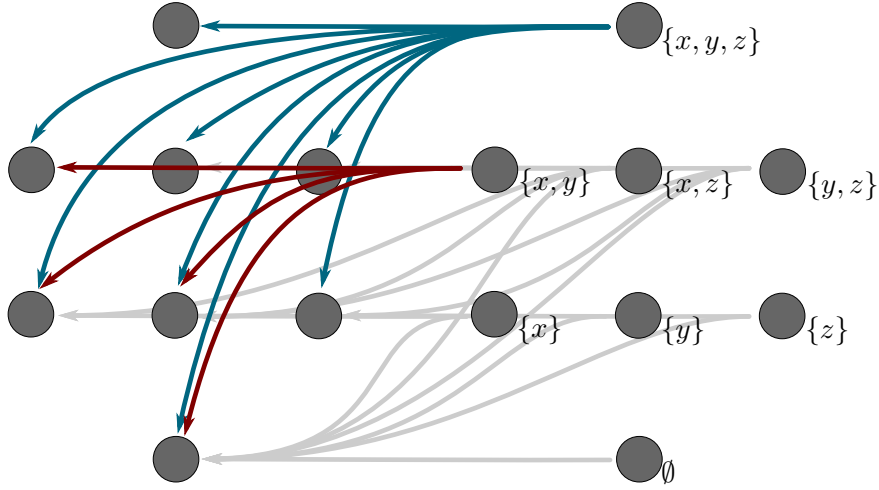

    \caption{The RBA from Example \ref{ex:RBA} as embedded into the CRBA generated by $W = \{x,y,z\}$. The blue arrows are the elements of $\B$ and the red the elements of $\textup{\textbf{R}}$.}
    \label{fig:ex_CRBA}
\end{figure}

%Reliably, decision makers (DMs) cannot envisage all aspects of the environments in which they make choices. Indeed, the pervasiveness of unawareness is self evident in the face of scientific, academic, and cultural discoveries: humans before the 20th century were unaware of the tenets of quantum mechanics, DNA, the internet, the proof of Fermat's last theorem, etc. 

\subsection{Modalities}

% and it is naturally assumed that $\rr(E) \subseteq \aa(E)$, manifesting the same implication of knowledge on awareness as is present in the Kripke models.

%Awareness related concerns notwithstanding, the state-space models clearly share much with \emph{Boolean Algebras with Operators}, or BAOs. A BAO is a Boolean algebra, $\B = \<B,\land,\lor,\neg,0,1\>$ along with an operator $f:B\to B$, such that $f(1) = 1$ and $f(X\land Y) = f(X) \land f(Y)$. These two conditions embody the same restrictions as necessitation and the distribution (\textbf{K}) axiom, and so BAOs serves as a sound and complete semantics for regular modal logics.
%
%Indeed, for a BAO, $(\B,f)$ there exists a Kripke frame, $\UF(\B)$, whose worlds are the ultrafilters of $\B$, and which
%models exactly the same validities. Correspondingly, for each Kripke frame $F$, there exists a BAO, $\B^F$, where the underlying algebra is the powerset of the worlds in $F$. Stringing these constructions together, we obtain the celebrated representation theorem of \cite{jonnson1952boolean} that $\B$ can isomorphically embedded in $\B^{\UF(\B)}$.
%
%
%
%If $\B$ is a boolean algebra based on the set $B$ then $\2$ is the boolean algebra based on the set 
%$$\{(X,Y) \in B \times B \mid Y \geq X\}.$$
%

To accommodate reasoning about knowledge or other modalities, we can enrich an RBA, $\RB$, with an operator, a function $\rr: RB \to RB$.  To capture the standard properties of modalities in normal modal logics, we require that $\rr$ respects meets and maps the top element to itself.

When discussing knowledge and awareness, the interpretation is as in the state-space models: $\rr(X)$ is the element representing knowledge of the element $X$. As such, we also require $1_{\rr(X)} = 1_X$, that knowledge of an element must be defined in the same domain as the element itself. From $\rr$ we can define $\aa$, representing awareness, as $\aa(X) = \rr(1_X)$.
The definition of $\aa$, in addition to ensuring that awareness is domain specific, also embodies a weakened form of necessitation: the agent knows all (and only) tautologies she is aware of. 

\begin{figure}
    \centering
    \begin{minipage}{.49\textwidth}
    \scriptsize
    \def\svgwidth{.8\columnwidth}
    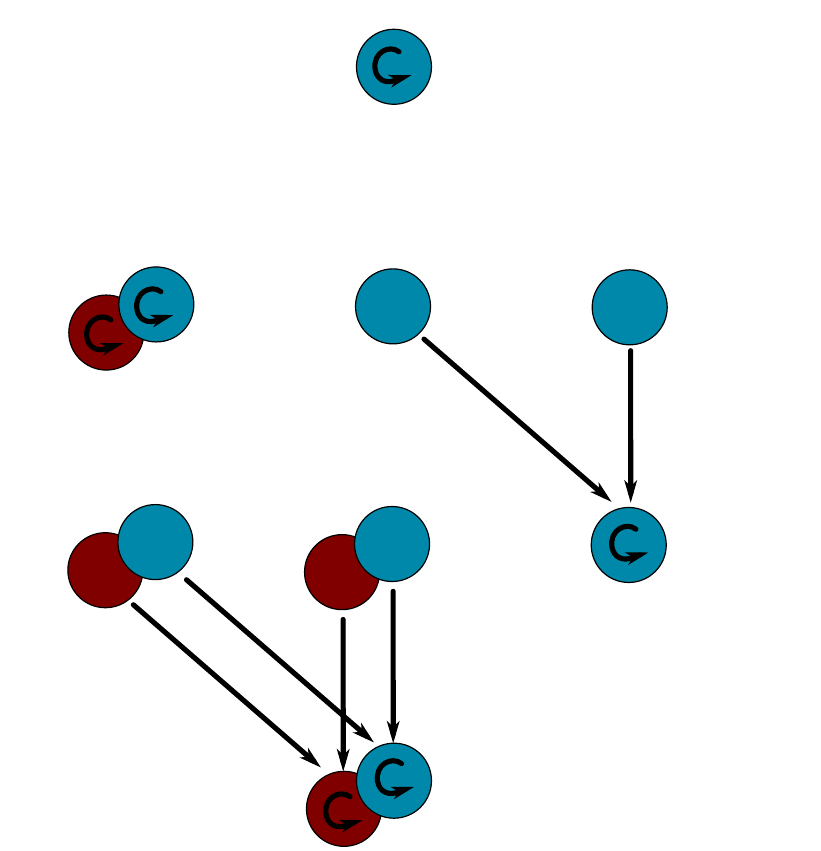

    \end{minipage}
    \hfill
        \begin{minipage}{.49\textwidth}
    \scriptsize
    \def\svgwidth{.8\columnwidth}
    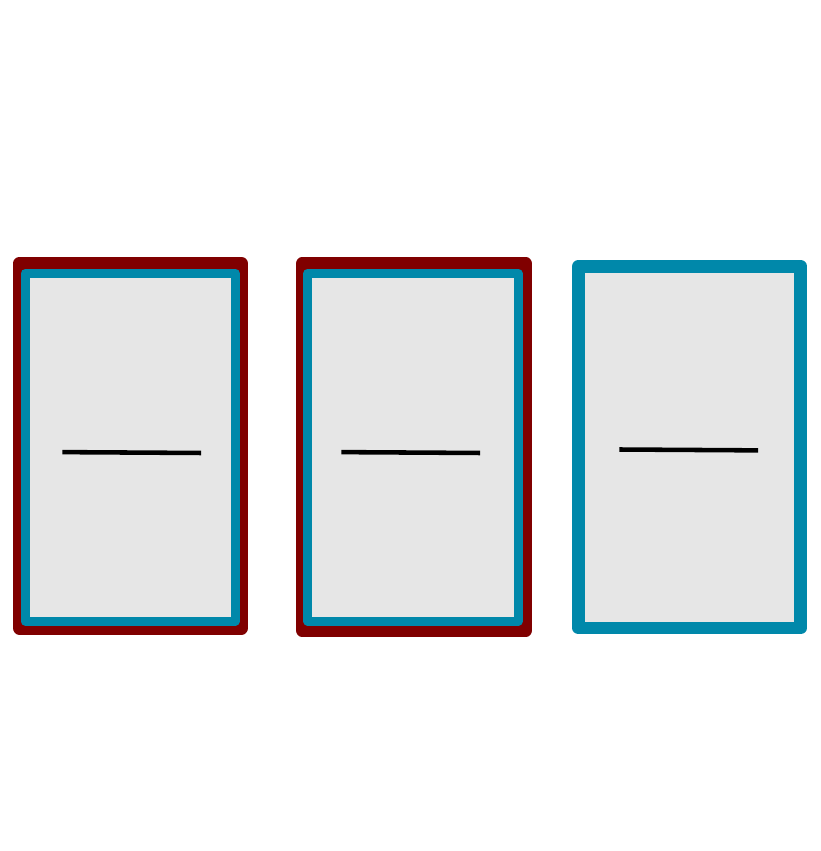

    \end{minipage}
    \caption{On the left, the RBA from Example \ref{ex:RBA} endowed with the operator $f^K$ from Example \ref{ex:fk}, as represented by the arrows. The right side shows a Kripke frame where the awareness sets are the languages generated by the propositions modeled and the accessibility relation, $R$, is partitional and given by the lines above the worlds. The (Boolean) algebra generated by the red worlds is $\textup{\textbf{R}}$, and by the blue, $\B$. The association of $\phi \mapsto \{\w | \w \models \phi\}$, produces the same model as in Example \ref{ex:RBA}, and $f^K$, from Example \ref{ex:fk}, then corresponds to $X \mapsto \{\w \mid R(\w) \subseteq X\}$.}
    \label{fig:ex_RBA_fk}
\end{figure}

\begin{ex}
\label{ex:fk}
We can reuse the RBA from Example \ref{ex:RBA} to capture an alternative story about awareness. 
Consider the proposition $p$ representing ``cryptographic protocol $x$ is insecure" and $q$ representing ``there is a quantum algorithm breaking protocol $x$." Associate $p$ to the event $X_B$ and $q$ to the event $X_R$. Then $\RB$ models the situation in which $p$ is always either true or false (since $\neg X_B = Y_B \lor Z_B$ so that $X_B \lor \neg X_B = 1)$ but $q$ is true or false only on the local domain where quantum computers exist (since  $\neg X_R = Y_R$ so that $X_B \lor \neg X_B = 1_R \neq 1)$.

Consider an agent whose awareness and knowledge are given by the operator $f^K: RB \to RB$ as given by $f^K(X_R) = f^K(Y_R) = 0_R$, 
$f^K(X_B) = f_K(Y_B) = 0_B$, $f^K(X_B\lor Z_B) = f_K(Y_B \lor Z_B) = Z_B$ and which coincides with the identity map everywhere else. This is visualized by the left side of Figure \ref{fig:ex_RBA_fk}. Then $f^A$ is simply the map $W \mapsto 1_W$.

Going back to our propositions $p$ and $q$, $f^K$ represents the epistemic state of affairs such that, if the agent is aware of quantum computers (i.e., is aware of $q$) then she is necessarily uncertain about the security of the protocol (i.e., does not know $p$). This is because the two elements resolving the truth of $p$ and representing awareness of $q$, namely, $X_R$ and $Y_R$, are known only at the bottom element. Conversely, if the agent who is unaware of $q$, she may be certain of $\neg p$; this is represented by $Z_B$.
\qed
\end{ex}

Propositions \ref{prop:comp} and \ref{prop:sound} show that modal RBAs are equivalent to awareness models (Kripke semantics for awareness logics) in exactly the same manner that modal algebras and Kripke frames are equivalent. For every modal RBA there is an awareness model that models the same theories and that constructed out of its ultrafilters. Conversely, for every awareness model there is a modal RBA constructed from the powerset of its worlds and that models the same theories. For example, the concrete RBA that embeds the RBA from Example \ref{ex:RBA}, itself visualized 
in Figure \ref{fig:ex_CRBA}, models the same theories as the Kripke frame shown on the right side of Figure \ref{fig:ex_RBA_fk}.

%Stringing these constructions together, we obtain a version of the celebrated representation theorem of \citet{jonnson1952boolean}, that every modal RBA can be isomorphically embedded into the (modal) concrete RBA constructed from its corresponding ultrafilter frame. Along the way this shows that the (propositionally generated) awareness logics considered by \citet{fagin1987belief} are complete and sound with respect to the class of modal RBAs. 

A shift in perspective shows that RBAs can capture other types of modal environments where truth is not absolute:

\begin{ex}
A runner is standing at the start line of the 800m race. Let $p$ denote the statement ``The runner is the winner of the 800m.'' It is reasonable to say that in the current state-of-affairs, $p$ is neither true nor false, but rather undetermined. This indeterminacy arises specifically because there is some further resolution of the state-of-affairs which resolves $p$ to be true and another that resolves it to be false. Now consider a spectator, who is sure of the current state-of-affairs and set $Kp$ as the statement ``The spectator knows $p$.'' Surely  $Kp$ is not true, but arguably, it is also philosophically sensible that it is not false either. That is, this is good reason to leave ``The spectator knows $p$'' undetermined just as $p$ is---the utility in not determining the truth assignment is that it distinguishes not knowing because of uncertainty about the state-of-affairs (the usual not knowing) with not knowing because the statement in question is undetermined and hence \emph{cannot be known}. Nonetheless, if the runner wins the race, so that indeterminacy resolves $p$ to be true, then the spectator will know $p$; likewise if $\neg p$ then $\neg K p$:  indeterminacy of the modal statement arises, just as before, because the state-of-affairs can be further determined towards $Kp$ or its negation.

This can be captured by the RBA consisting of the union of the elements of Boolean Algebras, $\B$, generated by the sets $\{X_B,\neg Y_B\}$, and $\textup{\textbf{R}}$, the trivial algebra $\{1_R, 0_R\}$. Let $1_R \land W_B= W_B$ for and $0_R \land W_B = 0_B$ for any $W_B \in \B$. Let $f^k$ be the identity map. 

By associating $p$ to the element $X_B$, the domain $\textup{\textbf{R}}$ represents the state-of-affairs where $p$ is not determined, and $\B$ the possible resolutions of this indeterminacy. In contrast to the case of differential existence / awareness, $p \lor \neg p$ \emph{is} determined even if $p$ is not, as its truth does not depend on how the indeterminacy of $p$ resolves. Thus, the state-of-affairs where $p$ is not determined is not specified by $\textup{\textbf{R}}$ alone but rather by the ways that $\textup{\textbf{R}}$ can can be extended to all of $\RB$: the event $1_R$ can be extended to a `consistent and complete set\footnote{Technically, we are extending ultra-filters of a domain to ultra-filters of $\RB$.} of elements' of $RB$ in two ways: $\{1_R,1_B,X_B\}$ and $\{1_R,1_B,\neg X_B\}$. Now $p$ is undetermined at $1_R$ because neither $X_B$ (the element associated to $p$) nor $\neg X_B$ (the element associated to $\neg p$) appear in every extension. $p \lor \neg p$, on the other hand, being associated to $1_B$ which resides in the intersection of all extensions, is determined to be true.  The same holds for various modal formulae.
\end{ex} 

Proposition \ref{prop:poss} shows that this `extension process,' lifting an element to the intersection of all complete and consistent extensions of it, reconstructs possibility semantics. That is, from an RBA, we can construct a persistent and refinable Kripke model with vague worlds ordered by their level of determinacy. Informally, the dictum of \emph{persistence and refinability} ensures that, as in the example, indeterminacy of a formula $\phi$ arises exactly when there are is a further resolution making $\phi$ true and another making it false.

\subsection{Projection} 
\label{sec:proj}

The extant models of semantic awareness \cite{heifetz2006interactive, li2009information} and possibility semantics \cite{humberstone1981worlds} all explicitly define their domains and the ordering over them (state-spaces and projections for awareness models; partial resolutions and a ordering relation for for possibility semantics). By contrast, for an RBA, the local domains and the order over them arise as derived objects from the algebraic relations. Thus, when interpreting an RBA as representing some logical system, the structure of the system (i.e., the domains and their order) and the logical relationships between the formulae themselves, arise from he same algebraic restrictions.

For a concrete benefit of this vantage, consider the class of algebraic structures that generalize RBAs so as to allow $\neg(Y \land 1_X) \neq \neg Y \land 1_X$ (for $Y \geq X$). It follows from the results below (Theorem \ref{thm:relative}) that for such structures, local domains still exist, and are still Boolean algebras, but the canonical projection maps might not exist. Thus, we can see that the \emph{structural} assumption in models of semantic awareness regarding the existence of projection maps equates to a \emph{logical} assumption about the distributivity of negation with the local $\land$-identity.

\section{Relativized Boolean Algebras}
\label{sec:alg}

\subsection{Preliminaries / Definitions}

Call $\bm{A} = \<A,\land,\lor,\neg,0,1\>$ an algebra \emph{of Boolean similarity type} when $A$ is a set, $0,1 \in A$, $\land$ and $\lor$ are binary operations taking $A\times A \to A$, referred to as the \emph{meet} and \emph{join}, receptively, and $\neg$ is a unary operation taking $A \to A$ referred to as the \emph{complement}. A \emph{homomorphism} $h: \bm{A} \to \bm{A}'$, is a function $h: A \to A'$ that maps $h(1) = 1'$ and that respects the operations (i.e., $h(X\land Y) = h(X) \land h(Y)$, etc).

If $\bm{A}$ is an algebra whose elements are partially ordered by $\geq$, then a \emph{filter}, $u$, on the algebra $\bm{A}$ is a subset of $A$ such that (i) $1 \in u$, (ii) if $X \in u$ and $Y \geq X$ then $Y \in u$ (i.e., $u$ is an $\geq$-upset) and (iii) if $X,Y \in u$ then $X \land Y \in u$  (i.e., $u$ is meet closed). A filter is called \emph{proper} if $u$ is a proper subset of $A$ and \emph{strongly} proper it does not contain $X\land \neg X$ for any $X \in A$.

An \emph{ultrafilter} is a filter that (iv) is strongly proper and there is no strongly proper filter, $v$, on $A$ such that $u$ is a proper subset of $v$. Let $\F(\bm{A})$ and $\U(\bm{A})$ denote the set of filters and ultrafilters on $\bm{A}$.\footnote{In Boolean algebras or other structures where $X \land \neg X = 0$ for all $X$, there is no distinction between proper and strongly proper filters.}

Of special importance is the class of Boolean algebras (whose elements are generically referred to as $\B$) that satisfy the axioms of Boolean algebras (see for example \cite{sep-boolalg-math}), written here for convenience:
\begin{enumerate}[leftmargin=2cm]
\item[\mylabel{ba1}{\textsc{ba1}}] $\land$ and $\lor$ are associative, communicative, and distributive. 
\item[\mylabel{ba2}{\textsc{ba2}}] $X \lor \neg X = 1$
\item[\mylabel{ba3}{\textsc{ba3}}] $X \land \neg X = 0$.
\item[\mylabel{ba4}{\textsc{ba4}}] $X \lor 0 = X \land 1 = X$.
%\item[\mylabel{ba5}{\textsc{ba5}}]  $Y \lor (X \land Y) = Y \land (X \lor Y) = Y$.
%\item For all $X \neq 0$, there exists a unique $Y$ such that $Y \land X = 0$ and $Z \geq Y$ implies $Z \land X \neq 0$.
\end{enumerate}
Let $\mathsf{BA}$ denote the class of Boolean Algebras. The operations induce a partial ordering on $B$ via $Y \geq X$ iff $X \lor Y = Y$ iff $X \land Y = X$. It is well known that condition (iv) in the definition of an ultrafilter is, for Boolean algebras, equivalent to: for all $X \in B$ either $X \in u$ or $\neg X \in u$, but not both.

\subsection{Axioms and Charecterization}
\label{sec:RBAax}

An algebra of Boolean similarity type,  $\RB = \<RB,\land,\lor,\neg,0,1\>$,  is a relativized Boolean algebra if it satisfies the laws below. To expedite their description, set the following notation $1_X \equiv X \lor \neg X$, $0_X \equiv X \land \neg X$, and $Y \geq X$ iff $X \land Y = X$. Let $\pi_2(\RB) = \{ 1_X \mid X \in RB\}$. For any $X \in RB$ let $\pi_1(X) = \{Z \in RB \mid 1_Z = 1_X\}$. 

For a garden variety Boolean algebra, $X \lor \neg X  = 1$ for all elements $X$; 1 is the unique the element that represents that which is determinately true. An RBA, intended to capture the notion of relative truth, comprises multiple local domains which each entertain their own local notion of truth. Specifically, the join operation is \emph{relative}, in the sense its identity, i.e., $0_X$, will depend on the domain of the elements on which it is acting.\footnote{\label{ft:ref1}Although $1$ is globally the identify for $\land$, the meet operation is also relative in the following sense: $X \land \neg X = 0_X$ where $0_X$ is not necessarily equal to $0$. The asymmetry between what is preserved under the meet and under the join arises from the fact that both operations move in the same $\geq$-direction across domains, while they move the in opposite $\geq$-directions within a domain.} 

$\pi_2(\RB)$ collects these domains, as indexed by their top elements. For each $1_X \in \pi_2(\RB)$, the corresponding domain is $\pi_1(1_X)$: the set of all elements $Z$ whose operational identities coincide with those of $X$, i.e., such that $Z \lor \neg Z = 1_X$. Thus, definitionally, for all $Z \in \pi_1(1_X)$, we have $\pi_1(Z) = \pi_1(1_X)$.\footnote{Of course RBAs generalize BAs. A Boolean algebra is a RBA for which $\pi_2(\RB)$ is a singleton and $\pi_1(X) = RB$ for all $X$.}

Theorem \ref{thm:relative} will establishes that the operations of an RBA, axiomatized below, will obey the laws of Boolean algebras locally, within the set of elements which have the same operational identities. That is: $\pi_1(X)$ forms a Boolean algebra. Moreover, the set these local Boolean structures themselves form a semi-lattice on which the projection maps are homomorphisms. 

The axioms of RBAs are:

\begin{enumerate}[leftmargin=2cm]
\item[\mylabel{rb1}{\textsc{rb1}}] $\land$ and $\lor$ are associative, communicative, and distributive and with $\neg$ satisfy DeMorgan's laws. 
\item[\mylabel{rb2}{\textsc{rb2}}]  $X \land X = X \lor X = X \land 1 = \neg\neg X = X$.
\item[\mylabel{rb3}{\textsc{rb3}}]  $X \lor 1 
%= \neg X \lor 1 X \lor \neg X$.
= 1_X$.
\item[\mylabel{rb4}{\textsc{rb4}}]  $X \land 0 = 0$.
%\item[\mylabel{rb4}{\textsc{rb4}}]  $X \lor 0 = X \land 0 = 0$.
%\item $1 = \neg 0$.
%\item $X \land \neg X = X \land 0$
%\item[\mylabel{rb5}{\textsc{rb5}}]  $1_X \geq 1_Y$  implies $\neg(X \land 1_Y)= \neg X \land 1_Y$.
%\item[\mylabel{rb5}{\textsc{rb5}}]  $1_Y \geq 1_X$  and $1_Z \geq 1_X$
%imply  
%	\begin{enumerate}
%	\item $\neg(Y \land 1_X)= \neg Y \land 1_X$.
%	\item $(Y \land 1_X) \land (Z \land 1_X) =(Y \land Z) \land 1_X$.
%	\item $(Y \land 1_X) \lor (Z \land 1_X) =(Y \lor Z) \land 1_X$.
%	\end{enumerate}
%\item For all $X \neq 0$, there exists a unique $Y$ such that $Y \land X = 0$ and $Z \geq Y$ implies $Z \land X \neq 0$.
\end{enumerate}

Let $\mathsf{RBA}^\circ$ denote the class of weak relativized Boolean algebras---those structures that adhere to \eqref{rb1}---\eqref{rb4}. The class of relativized Boolean algebras, denoted $\mathsf{RBA}$, are the elements of $\mathsf{RBA}^\circ$ that also satisfy:
\begin{enumerate}[leftmargin=2cm]
\item[\mylabel{rb5}{\textsc{rb5}}]  $X \geq Y$  implies $\neg(X \land 1_Y)= \neg X \land 1_Y$.
\end{enumerate}

%Say that $X \geq Y$ iff $X \land Y = X$. 
Notice, unlike in Boolean Algebras, $X \geq Y$ (i.e., that $Y \land X = Y$) need not be equivalent to $Y \lor X = X$. 
It is immediate from \eqref{rb2} and \eqref{rb4} that $1 \geq X \geq 0$ for all $X$. % and so it is generally permissible that there exist elements such that $0 \geq X$ and $X \not\geq 0$.

The axioms \hyperref[rb1]{\textsc{rb1}}--\hyperref[rb1]{\textsc{4}} find immediate counterparts in the rules of Boolean algebras.\footnote{Where \eqref{rb3} is a relativized version of the classical rule, and implies that the domain of a disjunction be lower (as ordered by $\geq$) than its constituent clauses: see Lemma \ref{lem:wo}\ref{wo:lat}.}  \eqref{rb5} require some motivation. If the various local domains bare no relation to each other, it would suffice to represent each independently; thus the value of RBAs lies in their ability to model connections between these domains of discourse, reconciling the local and global views. If $X \geq Y$ then $X \land 1_Y$ can be thought of as the representation of $X$ within the domain $\pi_1(Y)$---indeed it is the $\geq$-closest element to $X$ inside $\pi_1(Y)$.

\eqref{rb5} ensures that this map $X \to X \land 1_Y$ respects negation, and therefore (given the rest of the structure of RBAs) is a homomorphism. 
Thus, as will be established by Theorem \ref{thm:relative}, the $\pi(X)$-relative notion of truth can be \emph{projected} onto $\pi(Y)$ via these maps (which, further, compose with one another).

Building to this result, take first the following facts about RBAs

%For a Boolean algebra, $X \lor \neg X = 1$ irrespective of $X$, so the join is \emph{global}, in the sense that the join of an element and its complement is greater than any other element. Likewise  $X \land \neg X = 0$ for all $X$. As the name betrays, for a relativized Boolean algebra, the meet and join operations are \emph{relative}, in the sense that their identities, i.e., $1_X$ and $0_X$, depend on the elements on which they are acting. 

\begin{lem}
\label{lem:wo}
The following are true for all $\RB \in \mathsf{RBA}^\circ$.

\begin{enumerate}[label=(\roman*), leftmargin=2cm]
\item $\geq$ is a weak order. 
\item If $X \geq Y$ and $X' \geq Y$ then $X \land X' \geq Y$ and $X \lor X' \geq Y$. 
\item\label{wo:order} If $X \geq Y$ then $1_X \geq 1_Y$ and $X \land 0_Y = 0_Y$. 
\item\label{wo:in} If $1_X \geq 1_Y$ then $X \land 1_Y \in \pi_1(Y)$. 
\item\label{wo:lat} $1_X \land 1_Y = 1_{X\land Y} = 1_{X \lor Y} = 1_X \lor 1_Y$.
\end{enumerate}
\end{lem}

Lemma \ref{lem:wo} characterizes the relation between the domains via the ordering of elements: \ref{wo:in} show that our notion of projecting is well founded, and  \ref{wo:lat} shows that the domains inherent the lattice structure from the elements themselves. These results position us to show that these domains indeed capture relative, and via the projections, interconnected, notions of truth.  

\begin{thm}
\label{thm:relative}
Let $\RB \in \mathsf{RBA}^\circ$. Then
	\begin{enumerate}
	\item For each $X$, $\pi_1(X) = \{Z \in RB \mid 1_Z = 1_X\}$ is a Boolean Algebra (with $1_X$ and $0_X$ as the top and bottom elements, and the inherited operations).
	\item $\RB \in \mathsf{RBA}$ (i.e., satisfies \eqref{rb5}) if and only if for all $X \geq Y$ the map $h^X_Y: Z \mapsto Z \land 1_Y$ is a homomorphism from $\pi_1(X)$  to $\pi_1(Y)$. Moreover, in such cases, if $X \geq Y \geq Z$, then $h^X_Y \circ h^Y_Z = h^X_Z$.
	\end{enumerate}
\end{thm}

\begin{proof}
(1) That $\pi_1(X)$ is closed under $\neg, \land,$ and $\lor$ is immediate. That these relations satisfy \eqref{ba1} follows from \eqref{rb1}. Let $Y \in \pi_1(X)$. Then $Y\lor \neg Y = Y \lor 1 = X \lor 1 = X \lor \neg X = 1_X$. Likewise, $Y\land \neg Y = \neg(Y \lor \neg Y) = \neg(X \lor \neg X) = X \land \neg X = 0_X$. So \eqref{ba2} and \eqref{ba3} hold. Finally, consider $Y \land 1_X$ and  $Y \lor 0_X$. We have $Y \land 1_X = Y \land (Y\lor 1) = (Y\land Y) \lor (Y\land 1) = Y \lor Y = Y$ and also $Y \lor 0_X = Y \lor (Y\land \neg Y) = (Y\lor Y) \land (Y\lor \neg Y) = Y \land 1_X = Y$, indicating \eqref{ba4}.

(2) Assume \eqref{rb5} holds, and let $X \geq Y$ and consider the map $h^X_Y$. Lemma \ref{lem:wo}\ref{wo:in} states that the image of $h^X_Y$ is indeed $\pi_1(Y)$. Now let 
$Z,Z' \in \pi_1(X)$. We have $\neg h^X_Y(Z) =  \neg(Z \land 1_Y) = \neg Z \land 1_Y = h^X_Y(\neg Z)$ by \eqref{rb5}. We have $h^X_Y(Z\land Z') = (Z \land Z') \land 1_Y = (Z \land 1_Y) \land (Z' \land 1_Y) = h^X_Y(Z) \land h^X_Y(Z')$. The case for $\lor$ is identical. So $h^X_Y$ is a homomorphism.
%$1_Y \geq 1_X$ and $1_Z \geq 1_X$. We have $\neg h_X(Y) =  \neg(Y \land 1_X)= \neg Y \land 1_X = h_X(\neg Y)$ by \eqref{rb5}. We have $h_X(Y\land Z) = (Y \land Z) \land 1_X = (Y\land 1_X) \land (Z \land 1_X) = h_X(Y) \land h_X(Z)$. The case for $\lor$ is identical. 

Next, let $X \geq Y$ and assume $h^X_Y$ is a homomorphism. Then $\neg(X \land 1_Y) = \neg h^X_Y(X) = h^X_Y(\neg X) =  \neg X \land 1_Y$. So \eqref{rb5} holds.

Finally, let $X \geq Y \geq Z$ and $W \in \pi_1(X)$. Then $h^Y_Z(h^X_Y(X)) = (W \land 1_Y) \land 1_Z = W \land (1_Y \land 1_Z) = W \land 1_Z = h^X_Z(W)$, where the penultimate equality arises from the fact that $1_Y \geq 1_Z$, so $1_Y \land 1_Z = 1_Z$.
\end{proof}

\subsection{Concrete Relativized Boolean Algebras}

Just as the powerset of a set forms the prototypical example of a $\mathsf{BA}$, $\mathsf{RBA}$s can also be given concrete instantiations as (a subset of) a powerset. A \emph{concrete} $\mathsf{RBA}$ based on a set $W$ has elements which are of the form $(A,B)$ where both $A$ and $B$ are subsets of $W$ and $A \subseteq B$. The operations are relative versions of the usual powerset Boolean algebra operations: for example $\neg(A,B) = ( B\setminus A, B)$.

In this section we will establish a version of Stone's representation theorem for $\mathsf{RBA}$s, showing that each $\mathsf{RBA}$ can be embedded into a concrete relative Boolean algebra. 

If $W$ is a set, let $\2$ denote $\<\{(A,B) \mid B \subseteq W, A \subseteq B\}, \land,\lor,\neg,(\emptyset,\emptyset),(W,W)\>$ with the operations being defined as follows:

\begin{enumerate}[leftmargin=2cm]
\item[\mylabel{crb1}{\textsc{neg}}] $\neg(A,B) = ( B\setminus A, B)$
\item[\mylabel{crb1}{\textsc{meet}}] $(A,B) \land (A',B') = (A \cap A' , B \cap B')$
\item[\mylabel{crb1}{\textsc{join}}] $(A,B) \lor (A',B') = ((A \cup A') \cap (B \cap B'), B \cap B')$
\end{enumerate}

Let $\mathsf{CRBA}$ be the class of all such algebras.
%For a non-trivial $\B$, $\bf{2B}$ is not a Boolean algebra; this is most obvious from the fact that $(X,Y) \lor \neg (X,Y) = (Y,Y) \neq (1,1)$ for $Y \neq 1$.
It is easy to check that $\2$ is a relativized Boolean algebra, with $\pi_2(\2) = \{ (B,B) \mid B \subseteq W\} \cong \p(W)$ and $\pi_1(A,B) = \{(A',B) \mid  A' \subseteq B \} \cong \p(B)$ both being Boolean algebras.
%As seen below, this observation indicates that $\mathsf{2BA}$ provides an alternative (more complicated) algebraic semantics for classical logic.
Notice here that the ordering, $\geq$, is simply the product ordering on $\p(W)$: $(A,B) \geq (A',B')$ if and only if $A \supseteq A'$ and $B \supseteq B'$. Also notice that $1_{(A,B)} = (B,B)$ and $0_{(A,B)} = (\emptyset,B)$.

%Say that $\phi \in \L(\P)$ is \emph{valid} in $\mathsf{2BA}$, denoted $\mathsf{2BA} \models \phi$, if for all $\2 \in \mathsf{2BA}$ and all homomorphisms from $h: \L(\P) \to \2$ we have that $h(\phi) = (X,X)$ for some $X \in \B$.
%
%\begin{thm}
%$\phi$ is a theorem of classical logic if and only if $\mathsf{2BA} \models \phi$.
%\end{thm} 
%
%\begin{proof}
%For $B \in \mathsf{BA}$, let $h: \L(\P) \to \2$ and let $h(\phi) = (X,Y)$. 
%Define the homomorphism $h': \L(\P) \to \2(Y)$ via its generators:
%$$h': p \mapsto (\pi_1(h(p)) \cdot Y,Y),$$
%where $\pi$ is the projection onto the $i^{th}$ coordinate. 
%A simple induction argument shows that if $\psi$ is a subformula of $\phi$ then $h'(\psi) = (\pi_1(h(\psi)) \cdot Y,Y)$; this uses the fact that if $\psi$ is a subformula of $\phi$ then $\pi_2(h(\psi)) \geq Y$. 
%Since $\phi$ is a theorem of classical logic and $h'$ is a homomorphism to the Boolean algebra, $\2(Y)$, it must be that $h(\phi) = (Y,Y)$, and hence $\pi_1(h(\phi)) = Y$ as desired. 
%
%Completeness is even simpler: if $\phi$ is a non-theorem of classical logic, by the completeness of $\mathsf{BA}$ there exists a $\B \in \mathsf{BA}$ and a homomorphism  $h: \L(\P) \to \B$ such that $h(\phi) \neq 1$. The map from $\L(\P)$ to $2$ be defined by $h': \psi \mapsto (h(\psi),1)$ is clearly a homomorphism and $h'(\phi) = (h(\phi),1) \neq (1,1)$.
%\end{proof}

Theorem \ref{thm:embed}, below, shows that every $\RB \in \mathsf{RBA}$ can be embedded into the concrete $\mathsf{RBA}$ based on the set $F^\RB$. While the result is a reasonably straightforward generalization of Stone-line representation, it has some specific philosophical value in the present context. While is clearly possible to represent each local domain (being a classical Boolean algebra) as a concrete entity, it is much less clear that one could represent \emph{the connection between} various local notions of truth in a concrete way. Indeed, many extant models of local truth (e.g., \cite{humberstone1981worlds, heifetz2008canonical, halpern2013reasoning}) construct disjoint local state spaces and then explicitly and externally model the relation between them. By contrast, Theorem \ref{thm:embed} ensures it is always possible to represent via a single set, both the local domains and their relation to one another. 

Because $\pi_2(\RB)$ need not be a Boolean algebra, the filters we must work with are not ultrafilters. Towards this, if $u$ is a filter of $\RB$, then let $\pi_2(u) =  u \cap \pi_2(\RB)$ and $\pi_1(u,X) = u \cap \pi_1(X)$. With these definitions, we can 
can define the following filters on $\RB$, which may not be ultrafilters themselves but whose projections are either empty or are ultrafilters.
$$F^\RB = \{ u \in \F(\RB) \mid
% \pi_2(u) \in \U(\pi_2(\RB)),  
\pi_1(u,X) \in \U(\pi_1(X)), \text{ for all } X \in \pi_2(u)\}.$$
%The following Lemma establishes that we can extend any suitable filter to an element of $F^\RB$.  
Using these filters, we can construct our concrete representation. 

\begin{thm}
\label{thm:embed}
The map $h: \RB \to \textup{\textbf{2}}\p(F^\RB)$ defined by
$$X \mapsto (\{ u \in F^\RB \mid X \in u\}, \{ u \in F^\RB \mid 1_X \in u\})$$
is an injective homomorphism.
\end{thm}

\begin{proof}
Clearly $h(1) = (F^\RB,F^\RB)$. Take $h(X) = (A,B)$ and $h(Y) = (A',B')$. By \eqref{rb3}, $1_X = 1_{\neg X}$; we have $\{ u \in F^\RB \mid 1_{\neg X} \in u\} = B$. If $u \notin B$, then $\neg X \notin u$, and if $u \in B$ then $u \cap \pi_1(X)$ is an ultrafilter on $\pi_1(X)$. Hence for all $u \in B$, either $X \in u$ or $\neg X \in u$. This indicates that $h(X) = (B\setminus A,B)$ as desired.

Since a filter contains $X\land Y$ if and only if it contains $X$ and it contains $Y$, $\{ u \in F^\RB \mid X\land Y \in u\} = A \cap A'$. Further, by Lemma \ref{lem:wo}\ref{wo:lat}, $1_{X\land Y} = 1_X \land 1_Y$, so $\{ u \in F^\RB \mid 1_{X\land Y} \in u\} = B \cap B'$: $h(X\land Y) = (A \cap A', B\cap B')$.

Next, by Lemma \ref{lem:wo}\ref{wo:lat}, $1_{X\lor Y} = 1_X \land 1_Y$, so $\{ u \in F^\RB \mid 1_{X\lor Y} \in u\} = B \cap B'$, as well. Now: If $u \notin B\cap B'$, then $X\lor Y \notin u$, and if $u \in B$ then $u \cap \pi_1(X\lor Y)$ is an ultrafilter on $\pi_1(X\lor Y)$. It is well known that an ultrafilter on a Boolean algebra contains $X\lor Y$ if and only if it contains $X$ or it contains $Y$. This indicates that $h(X\lor Y) = ((A \cup A') \cap (B \cap B'), B \cap B')$, as desired.

Finally to see that $h$ is injective, assume $X \neq Y$. If $1_X \neq 1_Y$ then assume without loss of generality that $1_Y \not\geq 1_X$. So $\{1_Z \mid 1_Z \geq 1_X, Z \in RB\}$ is a strongly proper filter on $\RB$ and can, by Lemma \ref{lem:ultrafilterextension}, be extended to an element of $F^\RB$ that does not include $1_Y$. Thus $B \neq B'$.

If $1_X = 1_Y$, and (without loss if generality $Y \not\geq X$), then $\{Z \mid Z \geq X\}$ is a strongly proper filter on $\RB$ not containing $Y$. By Lemma \ref{lem:ultrafilterextension} again, we can extend to an element of $F^\RB$ that does not include $Y$. Thus $A \neq A'$. 
\end{proof}

\subsection{Modal $\RB$'s}
\label{sec:baos}

\renewcommand\rr{f}

As many of the interpretations of relative truth arise from a modal/relational structure, we introduce the a \emph{modal RBA}, a RBA equipped with an operator. Specifically, if $\RB \in \mathsf{RBA}$, then $(\RB,\rr)$ is a \emph{modal} relativized Boolean algebra, or MRBA, if 
$$\rr: RB \to RB$$ 
such that the following conditions hold
\begin{enumerate}[leftmargin=2cm]
\item[\mylabel{f1}{\textsc{f1}}]  $\rr(1) = 1$
%\item $\rr(Y,Y) = (Y,Y)$, and,
\item[\mylabel{f2}{\textsc{f2}}]  $\rr(X \land Y) = \rr(X) \land \rr(Y)$.
%\item[\mylabel{f3}{\textsc{f3}}]  $1_{\rr(X)} = 1_X$, %$\pi_2(\rr(X,Y)) = Y$,
\item[\mylabel{fD}{\textsc{fD}}]  $\rr(0_X) = 0_X$.
\end{enumerate}

Let $\mathsf{MRBA}$ denote the class of modal RBAs.

These conditions reflect the properties of normal modal logics: \eqref{f1} reflects our weakened form of necessitation: something which is tautological \emph{and always defined} is necessary / known; \eqref{f2} encodes the distributive property of normal modalities. As always, we have that \eqref{f2} implies that $\rr$ is monotone. Finally, \eqref{fD} instantiates the restriction to non-triviality as \ax{D} does in frame semantics---that not \emph{not everything} is known / necessary. 

\section{Models of Propositional Logic} 

For $\P$, a set of propositional variables, let $\L(\P)$ be the language defined by the grammar 
$$ \phi ::= p \mid 1 \mid \neg \phi \mid (\phi \land \phi) %\mid K\phi \mid A\phi
$$
where $p \in \P$. We employ the standard logical abbreviations: $0 \equiv \neg 1$, $(\phi \lor \psi) \equiv \neg(\neg \phi \land \neg \psi)$ and $(\phi \rightarrow \psi) \equiv (\neg \phi \lor \psi)$. %Given a set of formulae, $\Lambda \subseteq \L(\P)$ let $\P(\Lambda) \subseteq \P$ denote the set of propositional variables references by some $\phi \in \Lambda$, which can be formally defined recursively in the obvious way. 
For $\phi \in \L(\P)$ let $\P(\phi)$ collect those propositional variables which are subformula of $\phi$.

In an abuse of notation, we let $\L(\P)$ also denote the algebra of Boolean similarity type in which the base set is $\L(\P)$ itself and the meet, join, and complement operations and the top and bottom elements are denoted by their grammatical counterparts. This is the free algebra of Boolean similarity type generated by $\P$.

Say that $\phi \in \L(\P)$ is \emph{valid in $\RB$}, denoted $\RB \models \phi$,  if for all homomorphisms $h: \L(\P) \to \RB$ we have that $h(\phi) = 1_X$ for some $X \in RB$. Moreover say $\phi$ is \emph{valid in $\mathsf{RBA}$} (or just \emph{valid}), denoted $\mathsf{RBA} \models \phi$, if it is valid in $\RB$ for all $\RB \in \mathsf{RBA}$.

\begin{prop}
$\phi$ is a theorem of classical propositional logic if and only if $\mathsf{RBA} \models \phi$.
\end{prop} 

\begin{proof}
For $\RB \in \mathsf{RBA}$, let $h: \L(\P) \to \RB$ and let $h(\phi) = X$ for some classical validity, $\phi$. Note that for all subformula, $\psi$, of $\phi$, it must be that $1_{h(\psi)} \geq 1_X$ (this is the consequence of Lemma  \ref{lem:wo}).
%If $1_Y \not\geq 1_X$ then $1_{\neg Y} \not\geq 1_X$, $1_{Y\land  Z} \not\geq 1_X$ and $1_{Y\lor Z} \not\geq 1_X$
%
%Negation follows trivially from \eqref{rb3} and the other two cases from Lemma  \ref{lem:wo}\ref{wo:order}. 
%%Assume by way of contradiction that $((Y\land Z) \lor 1) \land (X\lor 1) = (X\lor 1)$. Then we have that $(Y \lor 1) \land (Z\lor 1) \land (X\lor 1) = (X\lor 1)$ iff $(Y \lor 1) \land(Y \lor 1) \land (Z\lor 1) \land (X\lor 1) = (Y \lor 1) \land(X\lor 1) \neq (X \lor 1)$, where the not-equals comes from our initial assumption. But this leads to a contradiction. An application of DeMorgan's law reduces the case for $\lor$ to the prior two cases.
Now, define the homomorphism $h': \L(\P) \to \pi_1(X)$ via:
$$h': p \mapsto 
\begin{cases}
(h(p) \land 1_X) \text{ if } 1_{h(p)} \geq 1_X \\
0_X \text{ otherwise.}
\end{cases}
$$
By Theorem \ref{thm:relative}, it must be that for all subformula, $\psi$, of $\phi$, that
$$h'(\psi) = (h(\psi) \land 1_X) = h^{h(\psi)}_X(h(\psi)) \in \pi_1(X).$$ 
Since $\phi$ is a theorem of classical logic and $h'$ is a homomorphism to $\pi_1(X) \in \mathsf{BA}$, it must be that $h'(\phi) = 1_X$, and hence $h(\phi) = 1_X$ as desired. 

Completeness follows from the fact that $\mathsf{BA} \subset \mathsf{RBA}$.
\end{proof}

\section{Awareness Semantics}

There are two interrelated methods of capturing awareness within a formal epistemology. First are the models that capture awareness semantically, where knowledge and awareness is understood in terms of the subsets of a set called a state-space \cite{dekel1998standard, heifetz2006interactive, li2009information}. Second are models that capture awareness syntacticly, where knowledge and awareness are understood in terms of statements about the world \citep{fagin1987belief, modica1994awareness, board2007object, halpern2009reasoning}.\footnote{There have also been several papers examining the connection/equivalence between extant models of the two approaches \citep{halpern2008interactive, heifetz2008canonical}.}

In state-space models, knowledge and awareness are represented by operators, $\rr$ and $\aa$ that map events (subsets of the state-space) to events. The event that an agent knows $E$ is $\rr(E)$; and that she is aware of event $E$ is $\aa(E)$. \citet{dekel1998standard} showed that under mild and intuitive conditions on these operators, the only possibility was being aware of everything or nothing. 
 
To circumvent this impossibility result, \citet*{heifetz2006interactive} (HMS) and \citet{li2009information} consider an ordered set of state-spaces. State-spaces higher in the ordering \emph{project} onto the lower spaces, in the sense that they are strictly more expressive. Then, roughly, an agent in state $\w$  is aware of those events which are in state-spaces lower in the ordering than the space containing $\w$. By considering multiple state-spaces, the definition of truth becomes inherently relative: 
%the union of an event and it complement need not be the whole space (i.e., need not be all the states in all the state-spaces). Thus, 
there are states in which neither an event nor its complement obtain. Nonetheless, when restricted to events in a particular state-space, behavior is classical.

Syntactic models of awareness, conversely, are necessarily contingent on a logical language, $\L$, with two modalities $A$ and $K$, respectively. The truth of formulas is then modeled via Kripke frames/models where at each possible world, $\w \in W$, the agent is aware of a subset of the underlying logical language, $\mathcal A(\w) \subseteq \L$, and considers a subset of the worlds $\r(\w) \subseteq W$, possible. Often each state $\w$ in endowed with only a subset of the full language, $\L(\w)$, and $\phi \in \L$ is modeled as true of false only at those states where $\phi \in \L(\w)$ \cite{modica1999unawareness, halpern2013reasoning, halpern2019partial}. Again, truth is relatively defined: the worlds where $\neg \phi$ is true is the  relative complement of those worlds where $\phi$ is true---relative to the worlds where it is defined. Validity is likewise relative; $\phi$ is considered valid if it is true in all states \emph{where it is defined}. 

\subsection{Awareness Models} 

Let $\L^{A,K}(\P)$ denote the extension of $\L(\P)$ to include the modalities $A$ and $K$: $\L^{A,K}(\P)$ is defined by
$$ \phi ::= p \mid 1 \mid \neg \phi \mid (\phi \land \phi) \mid A\phi \mid K\phi .$$

An \emph{ordered frame} is a pre-ordered set $(W,\geq)$ endowed with a serial binary relation, $\r$. We will set $\r(\w) = \{\w' \in \W \mid \w\r \w'\}$.
%and $\r(x)$ accordingly. 
%We require that if $y\in \r(x)$ then $y \geq x$. 
Although we refer to the elements of $W$ as worlds or states, note they will not have the standard interpretation of specifying the truth of all formulas, but will rather model only a subset of the language. 

An \emph{awareness model} for the language $\L^{A,K}(\P)$ is an ordered frame, $(W,\geq,\r)$ along with two functions, $L: \P \to \p(W)$ and $V: \P \to \p(W)$ such (i) $L(p)$ is $\geq$ upwards closed, and (ii) $V(p) \subseteq L(p)$ for all $p \in \P$.  Abusing notation let $\L^{A,K}(\w) = \L^{A,K}(L(\w))$ specify the language at world $\w$. It is the content of (i) that if $\w \geq \w'$ then $\L^{A,K}(\w) \supseteq \L^{A,K}(\w')$.

An awareness model $M = (W,\geq,\r,L,V)$ defines, at every $\w \in W$
the truth of all formula in $\L^{A,K}(\w)$. 
%By \eqref{m2} the truth of each formula is determined in some state. 
Truth is defined recursively via the operator $\models$ as 

\begin{itemize}
\item $\<M, \w\> \models p$  
 iff $\w \in V(p)$,
 \item $\<M, \w\> \models \neg \phi$
  iff $\<M, \w\>  \not\models \phi$,
 \item $\<M, \w\> \models  \phi \land \psi$
  iff  $\<M, \w\> \models \phi$ and $\<M, \w\> \models \psi$,
  \item $\<M, \w\> \models A\phi$ 
  iff for all $\w' \in \r(\w)$, $\phi \in \L^{A,K}(\w')$,
 \item $\<M, \w\> \models K\phi$ 
 iff for all $\w' \in \r(\w)$, $\<M, \w'\> \models \phi$.
\end{itemize}

For a model $M$, let $V(\phi) = \{\w \in W \mid \<M, \w\> \models \phi\}$ collect the worlds in which $\phi$ holds, and $L(\phi) = \{\w \in W \mid \phi \in \L^{A,K}(\w)\}$ the worlds where $\phi$ is defined. The reuse of $V$ and $L$ is desired, as $V$ and $L$ are extensions of the functions in the definition of $M$.

Call $\phi$ \emph{valid} in $M$ if it is true everywhere it is defined: if $V(\phi) = L(\phi)$. Call $\phi$ valid in the class of awareness models, denoted $\mathsf{AM} \models \phi$, if it is valid in all $M$.

The awareness models considered here are slightly different than those proposed by Fagin and Halpern \cite{fagin1987belief}.
For propriety, a proof of equivalence is provided in the \hyperref[pf:thm1]{appendix}. In the Fagin and Halpern approach, the language is state-invariant and the awareness of an agent is given by an explicit set of statements $\mathcal{A}(\w) \subseteq \L$. Then $\<M, \w\> \models A\phi$ iff $\phi \in A(\w)$. As evidenced by Proposition \ref{prop:genaware}, the present model is equivalent to the Fagin-Halpern approach (under \ax{AGP}) and it yields two two benefits: (i) it makes the proofs in the next sections more straightforward, and (ii) it makes explicit the relation between $A$ to $K$, as both arise from the same relation, $\r$. In contrast to the purely syntactic approach of assigning to each state a set of formula about which the agent is aware, here, the agent's awareness, like her knowledge, arises from the possibilities she envisages. 
 
Notice that Necessitation (from $\phi$ infer $K \phi$) is not sound in $\mathsf{AM}$. Indeed, consider a model $M$ and some valid $\phi$ which is not in $\L^{A,K}(\w')$ for some $\w'$. Then if $\w' \in \r(\w)$ with $\phi \in \L^{A,K}(\w)$ we have that $\<M, \w\>\models \phi$ (since $\phi$ was valid and defined at $\w$, but not $\<M, \w\>\not\models K\phi$ (since $\phi$ is not defined at $\w'$, hence $\<M,\w'\> \not\models \phi$). Necessitation is sound and (along with the other axioms) complete within the class of frames where $\r(\w) \subseteq \{\w' \in W \mid \w' \geq \w\}$. However, this class of models is remarkably boring as evidenced by the validity of $\phi \rightarrow A\phi$.

\subsection{Modal $\RB$'s and Awareness}
\label{sec:baos}

\renewcommand\rr{f^k}

Let $\mathsf{MRBA^k}$ denote the class of modal RBAs, $(\RB,\rr)$, that satisfy the following additional restriction on the operator $\rr$:
\begin{enumerate}[leftmargin=2cm]
\item[\mylabel{fk}{\textsc{fk}}]  $1_{\rr(X)} = 1_X$, %$\pi_2(\rr(X,Y)) = Y$,
\end{enumerate}

This conditions reflect the property specific to knowledge and awareness in relation to the elements where they are defined: \eqref{fk} states that knowledge (and awareness) of an element is defined exactly when the event itself is defined---it is not possible to know or not know something which does not itself exist. From $\rr$ we can define the additional map $\aa: RB \to RB$ via $\aa: X \mapsto \rr(1_X)$.

If $(\RB,\rr)$ is a MRBA and $h: \L(\P) \to \RB$ is a homomorphism we can extend $h$ to $h^+: \L^{A,K}(\P) \to \RB$ via (inductively) $h^+(A\phi) = \aa(h^+(\phi))$ and $h^+(K\phi) = \rr(h^+(\phi))$. Then say that $\phi \in \L^{A,K}(\P)$ is \emph{valid} in $\mathsf{MRBA^k}$, denoted $\mathsf{MRBA^k} \models \phi$, if for all $(\RB,\rr) \in \mathsf{MRBA^k}$ and all homomorphisms from $h: \L(\P) \to \RB$ we have that $h^+(\phi) = 1_X$ for some $X \in RB$.

\begin{thm}
\label{thm:sandc}
$\mathsf{MRBA^k} \models \phi$ iff $\mathsf{AM} \models \phi$.
\end{thm}

The proof of Theorem \ref{thm:sandc} is the conjunction of the following two propositions. Proposition \ref{prop:comp} constructs, for each awareness model, a corresponding MRBA (and homomorphism) such that for each formula, $(V(\phi),L(\phi)) = h^+(\phi)$. Then, in converse fashion, Proposition \ref{prop:sound} constructs an awareness model, for each $(\RB,\rr,h)$, such that $h^+(\phi) = (V(\phi),L(\phi))$.  

\subsection{Powerset MRBAs}

If $F = (W,\geq,\r)$ is an ordered frame, define the concrete MRBA, $(\2, f^{K,R})$ and
 $$f^{K,R}: (A,B) \mapsto (\{\w \mid R^K(\w) \subseteq A\}\cap B, B)$$
 Verifying that $f^{K,R}$ here defined satisfies \eqref{f1}, \eqref{f2}, and \eqref{fk} is straight forward. \eqref{fD} follows from the assumption that $\r$ is serial. 
 %Notice that from this definition we get that
 %$$\aa: (X,Y) \mapsto (\{x \mid R(x) \subset Y)\cap Y, Y).$$

 \begin{prop}
 \label{prop:comp}
 Let $F = (W,\geq,\r)$ and $M = (F,L,V)$ be an awareness model and take $h^M:\L(\P) \to \2$ to be the homomorphism defined by $h^M(p) = (V(p),L(p))$ then 
 $h^{M+}(\phi) = (V(\phi), L(\phi))$.
 \end{prop}

 \begin{proof}
This is done by induction of the structure of formulae. The base case is the definition of $h^M$, and the steps for $\land$ and $\neg$ are immediate. We show the inductive steps for $K\phi$ and $A\phi$:
 \begin{align*}
 h^{M+}(K\phi) &= f^{K,R}(h^{M+}(\phi))\\
&= f^{K,R}(V(\phi), L(\phi))\\
&= (\{\w \mid R(\w) \subset V(\phi)\} \cap L(\phi), L(K\phi)) \\
&= (\{\w \mid \w' \in R(\w) \implies \<M,\w'\>\models \phi, \phi \in \L^{A,K}(\w)\}, L(K\phi)) \\
&= (V(K\phi), L(K\phi)).
 \end{align*}
The second equality is our inductive hypothesis. For awareness:
 \begin{align*}
 h^{M+}(A\phi) &= f^{A,R}(h^{M+}(\phi))\\
  &= f^{K,R}(h^{M+}(\phi \lor \neg\phi))\\
&= f^{K,R}(V(\phi \lor \neg\phi), L(\phi))\\
&= (\{\w \mid R(\w) \subset V(\phi \lor \neg\phi)\} \cap L(\phi), L(A\phi)) \\
&= (\{\w \mid \w' \in R(\w) \implies \<M,\w'\>\models \phi \lor \neg \phi, \phi \in \L^{A,K}(\w)\}, L(A\phi)) \\
&= (\{\w \mid \w' \in R(\w) \implies \phi \in \L^{A,K}(\w')\cap \L^{A,K}(\w)\}, L(A\phi)) \\
&= (V(A\phi), L(A\phi)).
 \end{align*}
 where again the third equality is our inductive hypothesis and the penultimate inequality from the fact that each state models the tautologies in its language, that is, $\<M,\w\>\models \phi \lor \neg \phi$ if and only if $\phi \in \L^{A,K}(\w)$.
 \end{proof}
 
Proposition \ref{prop:comp} proves that $\mathsf{MRBA^k} \models \phi$ implies $\mathsf{AM} \models \phi$. To see this, notice that if $\mathsf{MRBA^k} \models \phi$ then for every $(\RB,\rr)$, and for every homomorphism $h: \L(\P) \to \RB$, we have $h^+(\phi) = 1_X$ for some $X \in RB$. In particular, for each model $(F,L,V)$ this is true for $\2$ and $h^M: p \mapsto (V(p),L(p))$. Thus, Proposition \ref{prop:comp} requires that $V(\phi) = L(\phi)$. Since this holds for all models, we have that $\phi$ is valid in $\mathsf{AM}$.

\subsection{Ultrafilter Frames}

In dual fashion, the next Proposition shows that $\mathsf{AM} \models \phi$ implies $\mathsf{MRBA^k}\models \phi$ by constructing an awareness model for each $(\RB,\rr,h)$ that yield the same validities. As usual, the worlds will be sets of ultrafilter like objects. By the reasoning outlined in Section \ref{sec:alg}, we consider the filters
$$F^\RB = \{ u \in \F(\RB) \mid
% \pi_2(u) \in \U(\pi_2(\RB)),  
\pi_1(u,X) \in \U(\pi_1(X)), \text{ for all } X \in \pi_2(u)\}.$$
 
Then, if $(\RB,\rr) \in \mathsf{MRBA^k}$, define the \emph{ultrafilter frame} as $(F^\RB,\geq^\RB,\r^\RB)$ where $u \geq^\RB v$ iff $\pi_2(u) \supseteq \pi_2(v)$ and $u\r^\RB v$ iff $\rr(X) \in u$ implies $X \in v$. 
 
\begin{prop}
\label{prop:sound}
Let $h: \L(\P) \to \RB$ be a homomorphism and $h^+$ its extension to $\L^{A,K}(\P) \to \2$. Let $M = (F^\RB,\geq^\RB,\r^\RB, L^h, V^h)$ where $L^h(p) = \{u \in F^\RB \mid 1_{h(p)}\in u\}$ and $V^h(p) = \{u \in F^\RB \mid h(p) \in u\}$. Then for all $\phi \in \L^{A,K}(\P)$, 
$$L^h(\phi) =  \{u \in F^\RB \mid 1_{h^+(\phi)} \in u\},$$
and
$$V^h(\phi) = \{u \in F^\RB \mid h^+(\phi) \in u\}.$$
 \end{prop}
 
\begin{proof}

As always, the proof is by induction on the structure of formula. This is straightforward  with the help of the following lemma:
%, an obvious reincarnations of Lemmas \ref{lem:aextend} and \ref{lem:kextend}.

% \begin{lem}
% \label{lem:ultra}
% Let $u \in F^\RB$ and $\aa(X) \notin u$. Then there exists some $v \in F^\RB$ such that $\rr(X') \in u$ implies $X' \in v$ 
% and $X \notin v$.
% \end{lem}
% 
% \begin{subproof}{$\bigstar$}
% Define $v^- = \{ X' \in RB \mid \rr(X') \in u\}$. By \eqref{f2}, $1 \in v^-$, by \eqref{f3} $v^-$ is an upset and is closed under meets, hence $v^- \in \F(\RB)$. By \eqref{fD} $0_Y \notin v^-$ for any $Y \in RB$; so $v^-$ is strongly proper. Lemma \ref{lem:ultrafilterextension} allows us to extend $v^-$ to  $v \in F^\RB$ such that $\pi_2(v) = \pi_2(v^-)$. Since $\aa(X) = \rr(1_X) \notin u$,  $1_X \notin \pi_2(v^-) = \pi_2(v)$. It must be that $X \notin v$.
% \end{subproof}
% 
  \begin{lem}
  \label{lem:ultra2}
 Let $u \in F^\RB$. Then $\r(u) \subseteq \{v \in F^\RB \mid X \in v\}$ iff $\rr(X) \in u$. 
 \end{lem}
 
 \begin{subproof}{$\bigstar$}
 The if direction is immediate given the definition of $\r$. We show the only if via its contrapositive: if $\rr(X) \notin u$ then there exists some $v \in F^\RB$ such that $\rr(X') \in u$ implies $X' \in v$ and $X \notin v$.
 
Define $v^- = \{ X' \in RB \mid \rr(X') \in u\}$. By assumption $X \notin v^-$. By \eqref{f1}, $1 \in v^-$, by \eqref{f2} $v^-$ is an upset and is closed under meets, hence $v^- \in \F(\RB)$. By \eqref{fD} $0_Y \notin v^-$ for any $Y \in RB$; so $v^-$ is strongly proper. Lemma \ref{lem:ultrafilterextension} allows us to extend $v^-$ to  $v \in F^\RB$ such that $\pi_2(v) = \pi_2(v^-)$ and with $X \notin v$.
%Since $\aa(X) = \rr(1_X) \notin u$,  $1_X \notin \pi_2(v^-) = \pi_2(v)$. It must be that 
 \end{subproof}

We first show that for all $\phi$, $L^h(\phi) = \{v \in F^\RB \mid 1_{h^+(\phi)} \in v\}$. This is by induction of the complexity of $\phi$. The base case is the definition of $L$. We will show the cases for $\land$ and $K$ (negation is trivial and the argument for $A$ is exactly the argument for $K$).

Let $u \in L^h(\phi\land \psi)$. So, $\phi \land \psi \in \L(u)$ iff $\phi \in \L(u)$ and $\psi \in \L(u)$. By the inductive hypothesis, this is iff $1_{h^+(\phi)} \in u$ and $1_{h^+(\psi)} \in u$. Since $u \in F^\RB$, this is iff $1_{h^+(\phi)} \land 1_{h^+(\psi)} = 1_{h^+(\phi) \land h^+(\psi)} = 1_{h^+(\phi \land \psi)} \in u$.

Let $u \in L^h(K\phi)$. So, $K\phi \in \L(u)$ iff $\phi \in \L(u)$. By the inductive hypothesis, this is iff $1_{h^+(\psi)} = 1_{\rr(h^+(\psi))} = 1_{h^+(K\psi)} \in u$. Where the second equality is via \eqref{f1}.

Next, we show that for all $\phi$, $V^h(\phi) = \{v \in F^\RB \mid h^+(\phi) \in v\}$. Again, this is by induction of the complexity of $\phi$, and again, we will just show the interesting steps: Let $u \in V^h(K\phi)$. So, $\<M,u\> \models K\phi$; iff for all $v \in \r(u)$, $\<M,v\> \models \phi$. By the inductive hypothesis, this is iff $\r(u) \subseteq \{v \in F^\RB \mid h^+(\phi) \in v\}$, which, by Lemma \ref{lem:ultra2} is iff $\rr(h^+(\phi)) = h^+(K\phi) \in u$.

Let $u \in V^h(A\phi)$. So, $\<M,u\> \models A\phi$; iff for all $v \in \r(u)$,  $v \in L^h(\phi)$. By the previous part of the proof (concerning $L^h$), this
is iff $\r(u) \subseteq \{v \in F^\RB \mid 1_{h^+(\phi)} \in v\}$, which, by Lemma \ref{lem:ultra2} is iff $\rr(1_{h^+(\phi)}) = \aa({h^+(\phi)}) = h^+(A\phi) \in u$.
\end{proof}

To see that Proposition \ref{prop:sound} shows that $\mathsf{AM} \models \phi$ implies $\mathsf{MRBA^k} \models \phi$ (and thus completes the proof of Theorem \ref{thm:sandc}), let $\phi$ be valid in $\mathsf{AM}$ and pick your favorite $\mathsf{MRBA^k}$, $(\RB,\rr)$ and homomorphism $h: \L(\P) \to \RB$. Then in particular, $\phi$ is valid in $M = (F^\RB,\geq^\RB,\r^\RB, L^h, V^h)$, meaning $V^h(\phi)=L^h(\phi)$. Proposition \ref{prop:sound} then indicates that 
$$\{u \in F^\RB \mid h^+(\phi) \in u\} = \{u \in F^\RB \mid 1_{h^+(\phi)} \in u\},$$
which can only be true if $h^+(\phi) = 1_{h^+(\phi)}$, indicating validity in $(\RB,h)$.

\section{Possibility Semantics}
\label{sec:poss}

The above discussion focused on the interpretation of local domains as representing differential \emph{existence} of propositions. That is, in different domains, different things exist, and truth is relative to this existence (for the case of an unaware agent, existence might refer to what exists in the states she considers possible in her internal representation of the world).  Another reason to have local domains of truth would be differential \emph{determinacy} of propositions: states can be vague so that not every proposition is determined. This was the motivation of Humberstone \cite{humberstone1981worlds}, for the invention of possibility semantics: essentially possible world semantics where each world is a \emph{partial} resolution in the sense that it may leave some propositions indeterminate. These worlds are partially ordered by vagueness so that worlds higher in the ordering refine those below. Formally:

Let $\L^{\square}(\P)$ denote the extension of $\L(\P)$ to includes the modality $\square$. A possibility model is an ordered frame $(W, \geq, \r)$ along with a \emph{partial} function $\mathcal{V}: \P \times W \to \{\textbf{T},\textbf{F}\}$ that jointly satisfy:

%\begin{enumerate}[leftmargin=2cm]
%\item[\mylabel{per}{\textsc{persistence}}] \begin{enumerate}
%\item For all $p \in \P$ and $\w \in W$ if  $V(p,\w)$ is defined then $V(p,\w') = V(p,\w)$ for all $\w' \geq \w$.
%\item For all $\w,\w',\nu \in W$ if $\w' \geq \w$ and $\w'\r\nu$ then $\w\r\nu$
%\end{enumerate}
%\item[\mylabel{ref}{\textsc{refinability}}] \begin{enumerate}
%\item For all $p \in \P$ and $\w \in W$ if  $V(p,\w)$ is undefined then there exists a $\nu,\geq \w$ such that $V(p,\nu) = \textbf{T}$ and a a $\nu',\geq \w$ such that $V(p,\nu') = \textbf{F}$.
%\item For $\w,\nu \in W$, if $\w\R\nu$ then for some $\w' \geq \w$, $\w''\r\nu$ for all $\w'' \geq \w'$.
%\end{enumerate}
%\end{enumerate}
\begin{enumerate}[leftmargin=2cm]
\item[\mylabel{per}{\textsc{persistence}}] For all $p \in \P$ and $\w \in W$ if  $\mathcal{V}(p,\w)$ is defined then $\mathcal{V}(p,\w') = \mathcal{V}(p,\w)$ for all $\w' \geq \w$.
\item[\mylabel{ref}{\textsc{refinability}}] For all $p \in \P$ and $\w \in W$ if  $\mathcal{V}(p,\w)$ is undefined then there exists a $\nu\geq \w$ such that $\mathcal{V}(p,\nu) = \textbf{T}$ and a $\nu'\geq \w$ such that $\mathcal{V}(p,\nu') = \textbf{F}$.
\end{enumerate}
On these conditions Humberstone motivates 
\begin{quotation}
\small
Persistence is required because further delimitation of a possible state of affairs should not reverse truth values, but only reduce indeterminacies, and Refinability says that (for atomic formulae at least) such a reduction is possible in either of the relevant ways: this is a sort of 'principle of subdivision' for possibilities. 
\end{quotation}
%The (b) conditions extend these same principles to to modal formula, ensuring that for any $\phi \in \L^{\square}(\P)$, if the truth value of $\phi$ is defined at $\w$ then it identically defined for all $\geq$-greater states, and if the truth value is not defined then it can still be resolved, at some $\geq$-greater state, to either value.

Truth in a possibility model is defined recursively via the operator $\models$ as 

\begin{itemize}
\item $\<M, \w\> \models p$  
 iff $\mathcal{V}(p,\w) = \textbf{T}$,
 \item $\<M, \w\> \models  \phi \land \psi$
  iff  $\<M, \w\> \models \phi$ and $\<M, \w\> \models \psi$,
 \item $\<M, \w\> \models \neg \phi$
  iff for all $\w' \geq \w$, $\<M, \w'\>  \not\models \phi$,
 \item $\<M, \w\> \models \square\phi$ 
 iff for all $\w' \in \r(\w)$, $\<M, \w'\> \models \phi$.
\end{itemize}

In possibility model, the truth of a formula is relative to which atomic formula have been determined and so is not always defined. Notice that unlike the awareness models, where $p \lor \neg p$ is not defined when $p$ is not, in possibility models, the former is always determined (to be true) since it is independent of knowing the truth of $p$. Indeed, non-modal classical validities will be true at every state of a possibility model. 

Call a possibility model \emph{modally consistent} if for all $\phi \in \L^{\square}(\P)$, (i) if $\<M, \w\> \models \phi$ then for all $\w' \geq \w$, $\<M, \w'\> \models \phi$ and (ii) if $\<M, \w\> \not\models \phi$ and $\<M, \w\> \not\models \neg\phi$ then there exists some $\nu, \nu' \geq \w$ such that 
$\<M, \nu\> \models \phi$ and $\<M, \nu\> \models \neg\phi$. A modally consistent model is one which meets the analogs of persistence and refinability for all formulae. Clearly, persistence and refinability ensure these requirements are met for propositional formulae, and the semantics (in particular the rule for negation), ensures the structure propagates to all non-modal formulae. Modal formulae, on the other hand, need not generally conform---Humberstone ensures modal consistency by requiring joint conditions on $\geq$ and $\r$, namely  
\begin{enumerate}[leftmargin=2cm]
\item[\mylabel{perbox}{\textsc{p}}] For all $\w,\w',\nu \in W$ if $\w' \geq \w$ and $\w'\r\nu$ then $\w\r\nu$
\item[\mylabel{refbox}{\textsc{r}}] For $\w,\nu \in W$, if $\w\R\nu$ then for some $\w' \geq \w$, $\w''\r\nu$ for all $\w'' \geq \w'$.
\end{enumerate}

\begin{ex}
\label{ex:poss0}
Consider $\L^{\square}(\{p\})$ and let $W = \{\w_0,\w_{1a},\w_{1b}\}$ with $\w_{1a} \geq \w_0$, $\w_{1b} \geq \w_0$ and all the reflexive relations. Let $\r$ be the relation such that $\r(\w_{0}) = W$, $\r(\w_{1a}) = \w_{1a}$, $\r(\w_{1b}) = \w_{1b}$. Finally, let $V(p,\w_{1a}) = \textbf{T}$ and  $V(p,\w_{1b}) = \textbf{F}$. This model is modally consistent. In particular, $\square p$ is true at $\w_{1a}$ false at $\w_{1b}$ and undefined at $\w_0$. However, it does not satisfy \eqref{refbox} since $\w_0 \r \w_0$ but there is no $\w$ such that for all $\w' \geq \w$, $\w' \R \w_0$. 
\qed
\end{ex}

The next result shows that RBAs encode possibility semantics and, in particular, can be used to construct modally consistent models that do not necessarily adhere to \eqref{refbox} or \eqref{perbox}.

Let $(\RB,f^\square) \in \mathsf{MRBA}$ and $h: \L^{\square}(\P) \to \RB$ be a homomorphism. 
%Define 
%$$W^\dag = \bigcup_{X\in \pi_2(\RB)} \U(\pi_1(X))$$
%to be the (disjoint) union of all ultra-filters of all local domains of $\RB$.
As defined earlier, let $F^\RB$ denote the set of filters on $\RB$ with empty or ultrafilter projections onto local domains. 
For each $u \in F^\RB$ define $[u] = \{U \in \U(\RB) \mid u \subset U\}$ as the set of all ultrafilters containing $u$. Then define $\geq^\dag$ via $u \geq^\dag v$ if $[u] \subseteq [v]$ and $\r^\dag$ as $u \r^\dag v$ iff $f^\square(X) \in \bigcap [u]$ implies $X \in \bigcap[v]$. Finally, let $\mathcal{V}^\dag: \P \times W^\dag  \to \{\textbf{T},\textbf{F}\}$ be the partial function defined by $\mathcal{V}^\dag(p, u) = \textbf{T}$ iff $h(p) \in \bigcap [u]$ and  $\mathcal{V}^\dag(p, u) = \textbf{F}$ iff $h(\phi) \notin \bigcup [u]$.

\begin{prop}
\label{prop:poss}
Let $(\RB,f^\square) \in \mathsf{MRBA}$ and $h: \L^{\square}(\P) \to \RB$ be a homomorphism. 
Then $M = ((F^\RB, \geq^\dag, \r^\dag), \mathcal{V}^\dag))$ is a modally consistent possibility model and, in particular, $\<M, \w\> \models \phi$ iff $h_+(\phi) \in \bigcap [u]$.
\end{prop}

\begin{proof}
 To see that $M$ satisfies persistence and refinability: let $u' \geq^\dag u$, so that $[u'] \subseteq [u]$. Then $V(p,u) = \mathbf{T}$ iff $h(p) \in \bigcap [u] \subseteq  \bigcap [u']$. So $V(p,u') = \mathbf{T}$. Likewise, $\mathcal{V}^\dag(p,u) = \mathbf{F}$ iff $h(p) \notin \bigcup [u] \supseteq  \bigcup [u']$. So $\mathcal{V}^\dag(p,u') = \mathbf{F}$. $M$ is persistent. Now assume that $\mathcal{V}^\dag(p,u)$ is undefined. Then there must exist some $V, V' \in [u]$ such that $h(p) \notin V$ and $h(p) \in V'$. Further, $V,V' \in F^\RB$ and, since they are ultrafilters, are $\geq^\dag$-maximal. Clearly, $V \geq^\dag u$ and $V' \geq^\dag u$. $M$ is refinable. 

We next show that $\<M, \w\> \models \phi$ iff $h_+(\phi) \in \bigcap [u]$, which will in turn prove that $M$ is modally consistent. We show this via induction, the base-case being the definition of $\mathcal{V}^\dag$. 
	\begin{enumerate}
	\item[($\land$)] $\<M, u\> \models \phi \land \phi$ iff $\<M, u\> \models \phi$ and $\<M, u\> \models \phi$ iff (by the inductive hypothesis) $\{h_+(\phi), h_+(\psi)\} \subset \bigcap [u]$ iff $h_+(\phi) \land h_+(\psi) = h_+(\phi \land \psi) \in \bigcap [u]$ (since all ultrafilters are $\land$-closed). 
	\item[($\neg$)] $\<M, u\> \models \neg \phi$ iff for all $u' \geq u$, $\<M, u'\>  \not\models \phi$, iff (by the inductive hypothesis) for any $v$ such that $[v] \subseteq [u]$,  $h_+(\phi) \notin \bigcap [v]$, iff for all $U \in [u]$, $h_+(\phi) \notin U$, iff for all $U \in [u]$, $\neg h_+(\phi) = h_+(\neg \phi) \in U$ (since ultrafilters on Boolean algebras\footnote{Ah, you say, $U$ is an ultrafilter on an RBA \emph{not} a BA! But, as such, its projection on each local domain, a Boolean algebra, is an ultrafilter, and $X$ and $\neg X$ live in the same domian.} contain either $X$ or $\neg X$) iff $h_+(\neg \phi) \in \bigcap [u]$.
	\item[($\square$)] $\<M,u\> \models \square \phi$ iff for all $v \in \r(u)$, $\<M,v\> \models \phi$ iff $\r(u) \subseteq \{v \in F^\RB \mid h^+(\phi) \in \bigcap[v]\}$ (by the inductive hypothesis) which, by (a very slight alteration of) Lemma \ref{lem:ultra2} is iff $f^\square(h^+(\phi)) = h^+(\square\phi) \in \bigcap[u]$.
	\end{enumerate}
Now, finally, modal consistency follows from a recreation of the first paragraph of the proof. First: let $u' \geq^\dag u$, so that $[u'] \subseteq [u]$. Then $\<M, u\> \models \phi$ iff $h(\phi) \in \bigcap [u] \subseteq  \bigcap [u']$. So $\<M, u'\> \models \phi$.  
Now assume $\<M, u\> \not\models \phi$ and $\<M, u\> \not\models \neg\phi$. From the first assumption, we have that there exists some $V \in [u]$ such that $h(p) \notin V$. From the second assumption we have that there exists some $v' \geq^\dag u$ such that $\<M, v\> \models \phi$ and hence some $V' \in [v'] \subseteq [u]$ such that $h(p) \in V'$. As above, this shows $M$ is (modally) refinable. \end{proof}

\setcounter{ex}{2}

\begin{ex}[continued]
\label{ex:poss}
Let $RB$ consist of the union of the elements of Boolean Algebras, $\B$ (for blue) and $\textup{\textbf{R}}$ (for red), where $\B$ is generated by the sets $\{X_B,\neg Y_B\}$ and $\R$ is the trivial algebra $\{1_R, 0_R\}$. Let $1_R \land W_B= W_B$ for and $0_R \land W_B = 0_B$ for any $W_B \in \B$. Let $f^\square$ be the identity map. Take $h_+: \L^{\square}(\{p\}) \to \RB$ to be given by $h_+(p) = X_B$.

Then, $F^\RB$ consists of three filters:
$$ F^\RB = \big\{ u_0 := \{1_R\}, u_{1a} := \{1_R, 1_B, X_B\}, u_{1b} := \{1_R, 1_B, \neg X_B\}\big\}$$
In addition, $\U(\RB) = \{u_{1a}, u_{1b}\}$, so 
$[u_0] = \{u_{1a}, u_{1b}\}$, $[u_{1a}] = u_{1a}$, and $[u_{1b}] = u_{1b}$, and $\bigcap[u_0] = \{1_R, 1_B\}$, $\bigcap[u_{1a}] = \{1_R, 1_B, X_B\}$, and $\bigcap[u_{1b}] = \{1_R, 1_B, \neg X_B\}$. Finally, we have  $\mathcal V(p,u_{1a}) = \textbf{T}$ and  $\mathcal V(p,u_{1b}) = \textbf{F}$, whereas $\mathcal V(p,u_0)$ is undefined.

As such $\geq^\dag$ is reflexive and $u_{1a} \geq^\dag u_0$ and $u_{1b} \geq^\dag u_0$. Moreover, $\r^\dag$ is given by $\r^\dag(u_{0}) = F^\RB$, $\r^\dag(u_{1a}) = u_{1a}$, $\r^\dag(u_{1b}) = u_{1b}$. Thus, the model $((F^\RB, \geq^\dag, \r^\dag), \mathcal{V}^\dag))$ is, upto isomorphism, the modally consistent but not refinable model from earlier in Example \ref{ex:poss0}. 
\qed
\end{ex}

\section{Validity}
\label{sec:valid}

%In particular, for a given element $X$, $X \lor \neg X$ is not necessarily the top element, but rather the top element \emph{in the domain of $X$}. 
When RBAs are seen as an interpretation of a logic, it is not generally the case that formulae are globally defined---there are elements where neither $\phi$ nor $\neg \phi$ hold. While this is, essentially, the central feature of RBAs---and the feature that positions it as well suited to model differential existence, awareness,  and resolution of vagueness---it requires that we consider a non-standard notion of validity. 

Although we are unjustified to require $p \lor \neg p$ to be true \emph{in a state-of-affairs in which does not exist}, it seems nonetheless reasonable from the external prospective (outside any particular state-of-affairs) that $p \lor \neg p$ is a valid statement. Thus, we are guided to the notion of validity used here: a formula is valid if it is true so long as it is defined. Note that this is equivalent to defining a formula as valid if it is never false.\footnote{In a bivalent framework, this is equivalent to the classical definition of validity.}  \citet{halpern2008interactive} refer to this notion as \emph{weak validity}, and the classic notion of `aways-true-ness' as \emph{strong validity}. 
Weak validity has kicked around in many forms in the literatures on non-classical logics: multi-valued logic, partial logic, quasi-truth, etc \cite{lemmon1966algebraic, setlur1970equivalence, priest1979logic, bueno1996concept, kachi2002validity}. Weak validity as applied to state-spaces models in the discussion of awareness first appeared\footnote{\citet{modica1999unawareness} entertain a model with two types of states, objective and subjective, and introduce a notion of validity where a formula must be true at all `objective'  states. Since objective states model all formula, this corresponds, at least for non-modal formulae, to classical validity.} in \citet{heifetz2008canonical} and has been widely used since \citep{halpern2013reasoning, halpern2019partial}. It is worth noting that \citet{halpern2008interactive} propose a different resolution to the problem, choosing to augment the language to include a non-standard notion of implication, $\hookrightarrow$ to capture implication by undefined statements. Then $\phi \hookrightarrow \psi$ may be defined even when $\phi$ and $\psi$ are not: the authors take validity of $\phi$ to be the universal truth of $\neg(\phi \hookrightarrow 0)$.

Another, more roundabout, inspiration for weak validity arises from shift in prospective given by Section \ref{sec:poss}, which represents a different notion of truth. There we take $\phi$ to be `true' at $X$ not if $h(\phi)$ is in the filter generated by $X$, but rather, if $h(\phi)$ is in every ultra-filter containing $X$. This suggests the alternative notion of validity for RBAs: $\phi$ is valid in $\RB$ iff $h(\phi)$ is contained in every ultra-filter on $\RB$. Of course, under mild conditions, this new validity is the old validity:

\begin{prop}
Let $\RB \in \mathsf{RBA}$ and $\phi \in \L(\P)$. Let $\pi_2(\RB)$ have overlapping domains: that is for any $1_{X_1} \in \pi_2(\RB)$ and $Y \neq 0_Y$ we have $Y \land 1_X \neq 0_Z$ for any $Z\in RB$. Then $\RB \models \phi$ iff $h(\phi) \in \bigcap \U(\RB)$ for all homomorphisms $h:  \L(\P) \to \RB$.
\end{prop}

\begin{proof}

It suffices to show that $\bigcap \U(\RB) = \pi_2(\RB)$. That  $\bigcap \U(\RB) \subseteq \pi_2(\RB)$ is obvious, since if $X \notin  \pi_2(\RB)$ then $\neg X \neq 0_X$ and so, standard arguments show we can construct an ultrafilter containing $\neg X$.

So let $F \in \F(\RB)$ be a filter that does not contain $1_X$. Consider $v^- = \bigcup \{Y \land 1_X \mid Y \in F\}$. Since $F$ is a filter, by the overlapping domains property, this does not contain any local bottom elements, and by Theorem \ref{thm:relative}(2), $v^-$ does not contain $Z$ and $\neg Z$ for any $Z \in RB$. Hence we can extend $v^-$ to a filter $v$: we have $F \subsetneq F \cup \{1_X\} \subseteq v$, so $F$ was not an ultrafilter. Therefore $\pi_2(\RB) \subseteq  \bigcap \U(\RB)$.
\end{proof}

Therefore, our new prospective yields only change of interpretation, leaving weak validity intact.

%\section{Conclusion} 
%
%This paper proposes \emph{relativized Boolean algebras}, which serve as a semantics for relative truth. RBAs relax the restrictions of classical Boolean algebras so as to allow logical relations to be relative to a local domain. These domains are derived objects resulting from the algebraic identities defining RBAs and come equipped with a natural ordering and notion of projection. This paper shows that RBAs can be used to model both awareness and possibility. 

\newpage 
\appendix

\section{An example}

Let $RB$ consist of the union of the elements of Boolean Algebras, $\B$ (for blue) and $\textup{\textbf{R}}$ (for red), generated by the sets $\{Y_B,\neg Y_B\}$ and $\{X_R,\neg X_R\}$, respectively. Moreover, define the Boolean homomorphism $h^B_R: \B \to  \textup{\textbf{R}}$ defined by $Y_B \mapsto 1_R$. The operations on $\RB$, when restricted to either Boolean algebra,  coincide with the Boolean operations thereon. For $W_B \in \B$ and $W_R \in \textup{\textbf{R}}$, set $W_B \land W_R = h^B_R(W_B) \land W_R$, and $W_B \lor W_R = h^B_R(W_B) \lor W_R$. The top element is $1_B$ and the bottom is $0_R$. This algebra is visualized by Figure \ref{fig:ex_RBA_2}.

Notice that $h^B_R(\B) \neq \textup{\textbf{R}}$ but rather is the trivial $0-1$ algebra.

\begin{figure}   
    \centering 
    \begin{minipage}{.9\textwidth}
        \centering
    \def\svgwidth{.8\columnwidth}
    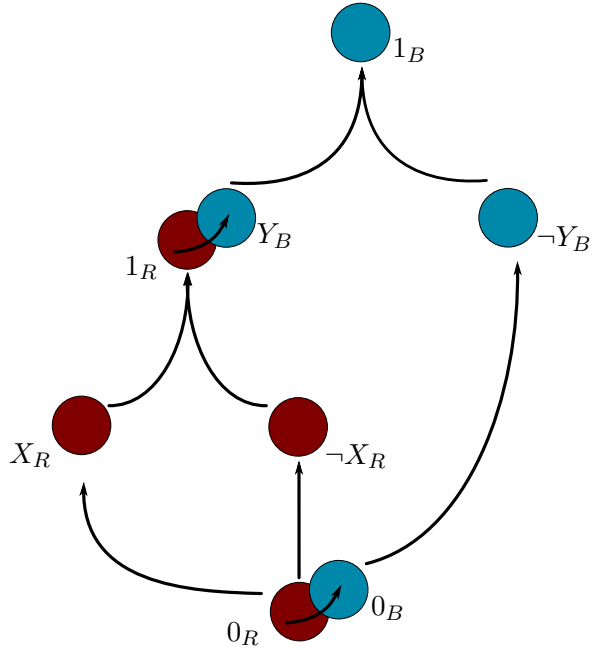

    \end{minipage}
    \caption{The RBA from Example 3. The arrows indicate the partial ordering $\geq$. The blue elements compose $\B$, and the red elements, $\textup{\textbf{R}}$.}
    \label{fig:ex_RBA_2}
\end{figure}

\section{Axiomatization of $\mathsf{AM}$}

Consider the following axioms over $\L^{A,K}(\P)$ all of which are standard, and whose merits and interpretations are discussed in the literature  cited in the introduction. 

\hspace{-4ex}
\begin{minipage}[t]{.5\textwidth}

\begin{centering}

{\bf Axioms:}

\end{centering}

\begin{description}

\item[{\textup K.}] $(K\phi\land K(\phi\rightarrow\psi))\rightarrow
K\psi$.

\item[{\textup D.}] $\neg K 0 $.

\item[{\textup T.}] $K\phi\rightarrow \phi$.

\item[{\textup 4.}] $K\phi\rightarrow KK\phi$.

%\item[{\textup 5.}] $\neg K\phi \rightarrow K\neg K\phi$.

\item[{\textup 5A.}] $(\neg K\phi \land A \phi) \rightarrow K\neg K\phi$.

\item[{\textup AGP.}] $A\phi \rightarrow
 A\psi$, for all $\psi \in \L^{A,K}(\P(\phi))$.%

\item[{\textup A0.}] $K\phi\rightarrow A\phi$.

\item[{\textup KA.}] $A \phi \leftrightarrow K A \phi$.

\end{description}
\end{minipage}
\begin{minipage}[t]{.5\textwidth}

\begin{centering}

\noindent {\bf Rules of Inference:}

\end{centering}

\begin{description}  
\item[{\textup MP.}] {F}rom $\phi$ and $\phi\rightarrow\psi$ infer $\psi$
(modus ponens).

\item[{\textup Sub.}] From $\phi$ infer all of its substitution instances.

\item[{\textup Nec$_{K}$.}] {F}rom $\phi$ infer $K \phi$.

\item[{\textup Nec$_{AK}$.}] {F}rom $\phi$ infer $A\phi \rightarrow K \phi$.

\end{description}
\end{minipage}

\bigskip

Let $\textbf{AX}$ denote the smallest logic containing the tautologies of propositional logic and $\ax{AGP}\cup\ax{K} \cup \ax{D}\cup \ax{A0}$ and which is closed under \ax{MP}, \ax{Sub}, and \ax{Nec$_{AK}$}. $\textbf{AX}$ is the axiom system considered in \cite{fagin1987belief} when awareness is generated by primitive propositions and when the accessibility relation is serial.\footnote{Our inclusion of \ax{D} will, as usual, specify those models where the accessibility relation is serial. There is no intrinsic problem considering a weaker logic without \ax{D} (and therefore without any restriction on the accessibility  relation), but to obtain a complete and sound axiomatization, we must replace it with a novel axiom: $K0 \rightarrow A\phi$.}

%The roles of $\ax{T},\ax{4}$, and $\ax{5A}$ are as in standard Kripke frames, axiomatizing the class of frames wherein $\r$ is reflexive, transitive, and Euclidean, respectively. \ax{KA} corresponds to the frames such that $\w' \in \r(\w)$ and $\w'' \in \r(\w')$ implies that $\w'' \geq \w'''$ for some $\w''' \in \r(\w)$. In particular, it is true in the class of transitive frames (i.e., under \ax{AX}$\cup$\ax{4}). 

\begin{prop}
\label{prop:genaware}
$\phi$ is a theorem of \ax{AX} iff $\mathsf{AM} \models \phi$.
\end{prop}

\label{pf:thm1}
\begin{subproof}{}
Soundness is straightforward. To show completeness, we follows the usual conical construction with a slight caveat. For the frame, let $W^c$ denote the set of all pairs $(\Gamma,\Q)$ where $\Gamma$ is a maximally consistent set of formula containing \textbf{AX}
in the language $\L^{A,K}(\Q)$ and $\Q \subseteq \P$. Order $W^c$ via $(\Gamma,\Q) \geq^c (\Gamma',\Q')$ iff $\Q \supseteq \mathbb{Q'}$. Construct the relations according to $(\Gamma,\Q)\r^{c}(\Gamma',\Q')$ iff $\{\phi \in \L^{A,K}(\P) \mid K\phi \in \Gamma\} \subseteq \Gamma'$ and $\{\phi \in \L^{A,K}(\P) \mid A\phi \in \Gamma\} \subseteq \L^{A,K}(\Q')$. Then to construct the canonical model, set $L^c(p) = \{ (\Gamma,\Q) \in \W^c \mid p \in \Q\}$ and $V^c(p) = \{ (\Gamma,\Q) \in \W^c \mid p \in \Gamma\}$.
An induction on the complexity of $\phi$ shows that, for all $(\Gamma,\Q)$, $\phi \in \Gamma$ iff $\<M^c, (\Gamma,\Q)\> \models \phi$. The only nontrivial steps, for $A\phi$ and $K\phi$, are direct consequences of the following lemmas:

\begin{lem}
\label{lem:aextend}
Fix $(\Gamma,\Q)$. If $A\phi \notin \Gamma$ then there exists a $(\Gamma',\Q')$ such that $(\Gamma,\Q)\r (\Gamma',\Q')$ and $\phi \notin \Q'$. 
\end{lem}

\begin{subproof}{$\bigstar$}
Since $\Gamma$ contains \ax{AGP}, we have that $\{\psi \mid A\psi \in \Gamma\} = \L^{A,K}(\Q')$ for some $\Q' 
\subset \P$. Since $A\phi \notin \Gamma$, $\phi \notin \L^{A,K}(\Q')$. Set $\Gamma^- = \{\psi \mid K\psi \in \Gamma\}$. By \ax{A0}, $\Gamma^- \subseteq \L^{A,K}(\Q')$, and by \ax{D}, $\Gamma^- \neq \L^{A,K}(\Q')$. Notice also that, by \ax{Nec$_{AK}$}, we have that $\Gamma^-$ contains all tautologies in $\L^{A,K}(\Q)$.  This allows for the standard argument that $\Gamma^-$ is a consistent set of formulas and can therefore be extended to a maximally consistent set, $\Gamma' \subset \L^{A,K}(\Q')$. $(\Gamma',\Q')$ is the desired world. 
\end{subproof}

\begin{lem}
\label{lem:kextend}
Fix $(\Gamma,\Q)$. If $K\phi \notin \Gamma$ then there exists a $(\Gamma',\Q')$ such that $(\Gamma,\Q)\r (\Gamma',\Q')$ and $\phi \notin \Gamma'$. 
\end{lem}

\begin{subproof}{$\bigstar\blacksquare$}
There are two cases to consider. First, if $A\phi \notin \Gamma$, then by Lemma \ref{lem:aextend}, there is an assessable world, $(\Gamma',\Q')$, such that $\phi \notin \L^{A,K}(\Q')$, and hence clearly, $\phi \notin \Gamma'$.

So assume that $A\phi \in \Gamma$. Since $\Gamma$ contains \ax{AGP}, we have that $\{\psi \mid A\psi \in \Gamma\} = \L^{A,K}(\Q')$ for some $\Q' \subseteq \P$. Then consider the set $\Gamma^- =  \{\neg\phi\} \cup \{\psi \mid K\psi \in \Gamma\}$. Since $\Gamma$ contains \ax{A0}, we have that $\Gamma^- \subseteq \L^{A,K}(\Q')$. As usual, $\Gamma^-$ can be extended to a maximally consistent set, $\Gamma'$ in $\L^{A,K}(\Q')$. Again, $(\Gamma',\Q')$ is the desired world. 
\end{subproof}
\end{subproof}

\section{Other Proofs and Lemmas}

\begin{proof}[Proof of Lemma \ref{lem:wo}]
Parts (i-iii) are immediate from definitions. 
\begin{enumerate}
\item[(iv)] $(X \land 1_Y)\lor 1 = 1_X \land ((Y \lor 1)\lor 1) = 1_X \land (Y \lor 1) = 1_X \land 1_Y = 1_Y$. 
\item[(v)]  $1_X \land 1_Y = (X \lor 1) \land (Y \lor 1 ) = (X \land Y) \lor 1 = 1_{X\land Y}$. Further $1_X \lor 1_Y = (X \lor 1) \lor (Y \lor 1 ) = (X \lor Y) \lor 1 = \neg(\neg X \land \neg Y) \lor 1 = (\neg X \land \neg Y) \lor 1 = (\neg X \lor 1) \land (\neg Y \lor 1) = (X \lor 1) \land (Y \lor 1) = 1_X \land 1_Y$, where the elimination of negations comes from the fact that $X \lor 1 = X \lor \neg X = \neg X \lor 1$ via commutativity and \eqref{rb3}.
\end{enumerate}
\end{proof}

\begin{lem}
\label{lem:ultrafilterextension}
Let $F \in \F(\RB)$ be strongly proper and let $X \notin F$. Then $F$ can be extended to $u \in F^\RB$ such that $\pi_2(F) = \pi_2(u)$ and $X \notin u$. 
\end{lem}

\begin{proof}[Proof of Lemma \ref{lem:ultrafilterextension}]
We will show that if $F \in \F(\RB)$ is strongly proper, then for all $1_X \in \pi_2(F)$, 
%if $X \notin F$ and 
$\neg X \notin F$ then $F' = \{Z \land Y \mid Z \geq X, Y \in F\}$ is in $\F(\RB)$ and is strongly proper and $\pi_2(F) = \pi_2(F')$ and $\neg X 
\notin F'$. This suffices, since we can then appeal to the usual Zornesque arguments, to choose a maximal element of the partial order of all extensions. 

%If $X \in F$ then $F' = F$, so assume that $X \notin F$. 
That $F'$ is upwards closed, contains $1$, and is closed under intersections is immediate. 
%Since for all $Z \land Y$ and $Z' \land Y'$ in $F'$, $(Z \land Y) \land (Z' \land Y') = (Z\land Z') \land (Y \land Y')$ and $(Z\land Z') \geq X$ and $(Y\land Y') \in F$, $F'$ is closed under intersections. 
Thus, we need only show that $F'$ is strongly proper. Assume to the contrary that $Z \land Y = 0_W$ for some $W \in RB$ with $Z \geq X$ and $Y \in F$. Then $\{1_Z,1_Y\} \in F$ and so to is $1_V = 1_Y \land 1_X$ and $Y \land 1_V$. Since $Z \land Y = 0_W$ we have also that\footnote{That $0_W \land 1_V = 0_V$ follows from the fact that $h_V$ given by Lemma \ref{thm:relative} is a homomorphism.} 
\begin{align}
\label{eq:its0}
(Z \land 1_V ) \land (Y \land 1_V) = 0_W \land 1_V = 0_V.
\end{align}
Now, $(Z \land 1_V ) \geq (X \land 1_V)$, the fact that $\pi_1(V) \in \mathsf{BA}$ and \eqref{eq:its0} requires that $\neg (X \land 1_V) \geq (Y \land 1_V)$. But, since $F$ was upwards closed, this requires that $\neg X \in F$, since $\neg X \geq \neg X \land 1_V = \neg (X \land 1_V)$ (by \eqref{rb5}). This contradicts our assumption.

Clearly, $\pi_2(F) \subseteq \pi_2(F')$, so to see the other direction, let $Z \land Y = 1_W$ for some $W \in RB$ with $Z \geq X$ and $Y \in F$. Then by Lemma  \ref{lem:wo}\ref{wo:lat}, $1_W = 1_Z \land 1_Y \geq 1_X \land 1_Y$ and so $1_W \in \pi_2(F)$. 
\end{proof}

\newpage

\bibliographystyle{../plainnatnourl.bst}
\singlespace
\bibliography{../AS.bib}

\end{document}

%% file: ex_RBA_1.pdf_tex
%% Creator: Inkscape inkscape 0.92.2, www.inkscape.org
%% PDF/EPS/PS + LaTeX output extension by Johan Engelen, 2010
%% Accompanies image file 'ex_RBA_1.pdf' (pdf, eps, ps)
%%
%% To include the image in your LaTeX document, write
%%   \input{<filename>.pdf_tex}
%%  instead of
%%   \includegraphics{<filename>.pdf}
%% To scale the image, write
%%   \def\svgwidth{<desired width>}
%%   \input{<filename>.pdf_tex}
%%  instead of
%%   \includegraphics[width=<desired width>]{<filename>.pdf}
%%
%% Images with a different path to the parent latex file can
%% be accessed with the `import' package (which may need to be
%% installed) using
%%   \usepackage{import}
%% in the preamble, and then including the image with
%%   \import{<path to file>}{<filename>.pdf_tex}
%% Alternatively, one can specify
%%   \graphicspath{{<path to file>/}}
%% 
%% For more information, please see info/svg-inkscape on CTAN:
%%   http://tug.ctan.org/tex-archive/info/svg-inkscape
%%
\begingroup%
  \makeatletter%
  \providecommand\color[2][]{%
    \errmessage{(Inkscape) Color is used for the text in Inkscape, but the package 'color.sty' is not loaded}%
    \renewcommand\color[2][]{}%
  }%
  \providecommand\transparent[1]{%
    \errmessage{(Inkscape) Transparency is used (non-zero) for the text in Inkscape, but the package 'transparent.sty' is not loaded}%
    \renewcommand\transparent[1]{}%
  }%
  \providecommand\rotatebox[2]{#2}%
  \ifx\svgwidth\undefined%
    \setlength{\unitlength}{481.88976378bp}%
    \ifx\svgscale\undefined%
      \relax%
    \else%
      \setlength{\unitlength}{\unitlength * \real{\svgscale}}%
    \fi%
  \else%
    \setlength{\unitlength}{\svgwidth}%
  \fi%
  \global\let\svgwidth\undefined%
  \global\let\svgscale\undefined%
  \makeatother%
  \begin{picture}(1,0.84705882)%
    \put(0,0){\includegraphics[width=\unitlength,page=1]{ex_RBA_1.pdf}}%
    \put(0.13851325,0.47736803){\color[rgb]{0,0,0}\makebox(0,0)[lb]{\smash{$1_R$}}}%
    \put(0.38292313,0.2349451){\color[rgb]{0,0,0}\makebox(0,0)[lb]{\smash{$Y_R$}}}%
    \put(0.12449009,0.23733498){\color[rgb]{0,0,0}\makebox(0,0)[lb]{\smash{$X_R$}}}%
    \put(0.36166721,0.01018749){\color[rgb]{0,0,0}\makebox(0,0)[lb]{\smash{$0_R$}}}%
    \put(0.30679902,0.51796803){\color[rgb]{0,0,0}\makebox(0,0)[lb]{\smash{$X_B {\lor} Y_B$}}}%
    \put(0.80966949,0.48265735){\color[rgb]{0,0,0}\makebox(0,0)[lb]{\smash{}}}%
    \put(0.30457564,0.27920986){\color[rgb]{0,0,0}\makebox(0,0)[lb]{\smash{$X_B$}}}%
    \put(0.53787065,0.28022618){\color[rgb]{0,0,0}\makebox(0,0)[lb]{\smash{$Y_B$}}}%
    \put(0.77430994,0.28337052){\color[rgb]{0,0,0}\makebox(0,0)[lb]{\smash{$Z_B$}}}%
    \put(0.53957605,0.04326675){\color[rgb]{0,0,0}\makebox(0,0)[lb]{\smash{$0_B$}}}%
    \put(0.98440561,0.07281602){\color[rgb]{0,0,0}\makebox(0,0)[lb]{\smash{}}}%
    \put(0.54268817,0.51223461){\color[rgb]{0,0,0}\makebox(0,0)[lb]{\smash{$X_B {\lor} Z_B$}}}%
    \put(0.77841752,0.51423566){\color[rgb]{0,0,0}\makebox(0,0)[lb]{\smash{$Y_B {\lor} Z_B$}}}%
    \put(0.54210097,0.75532839){\color[rgb]{0,0,0}\makebox(0,0)[lb]{\smash{$1_B$}}}%
    \put(1.06593377,0.04873738){\color[rgb]{0,0,0}\makebox(0,0)[lb]{\smash{}}}%
    \put(1.31905374,0.44492515){\color[rgb]{0,0,0}\makebox(0,0)[lb]{\smash{}}}%
    \put(-0.17826222,0.12936352){\color[rgb]{0,0,0}\makebox(0,0)[lt]{\begin{minipage}{0.26739331\unitlength}\raggedright \end{minipage}}}%
  \end{picture}%
\endgroup%

%% file: ex_CRBA.pdf_tex
%% Creator: Inkscape inkscape 0.92.2, www.inkscape.org
%% PDF/EPS/PS + LaTeX output extension by Johan Engelen, 2010
%% Accompanies image file 'ex_CRBA.pdf' (pdf, eps, ps)
%%
%% To include the image in your LaTeX document, write
%%   \input{<filename>.pdf_tex}
%%  instead of
%%   \includegraphics{<filename>.pdf}
%% To scale the image, write
%%   \def\svgwidth{<desired width>}
%%   \input{<filename>.pdf_tex}
%%  instead of
%%   \includegraphics[width=<desired width>]{<filename>.pdf}
%%
%% Images with a different path to the parent latex file can
%% be accessed with the `import' package (which may need to be
%% installed) using
%%   \usepackage{import}
%% in the preamble, and then including the image with
%%   \import{<path to file>}{<filename>.pdf_tex}
%% Alternatively, one can specify
%%   \graphicspath{{<path to file>/}}
%% 
%% For more information, please see info/svg-inkscape on CTAN:
%%   http://tug.ctan.org/tex-archive/info/svg-inkscape
%%
\begingroup%
  \makeatletter%
  \providecommand\color[2][]{%
    \errmessage{(Inkscape) Color is used for the text in Inkscape, but the package 'color.sty' is not loaded}%
    \renewcommand\color[2][]{}%
  }%
  \providecommand\transparent[1]{%
    \errmessage{(Inkscape) Transparency is used (non-zero) for the text in Inkscape, but the package 'transparent.sty' is not loaded}%
    \renewcommand\transparent[1]{}%
  }%
  \providecommand\rotatebox[2]{#2}%
  \ifx\svgwidth\undefined%
    \setlength{\unitlength}{564.09448819bp}%
    \ifx\svgscale\undefined%
      \relax%
    \else%
      \setlength{\unitlength}{\unitlength * \real{\svgscale}}%
    \fi%
  \else%
    \setlength{\unitlength}{\svgwidth}%
  \fi%
  \global\let\svgwidth\undefined%
  \global\let\svgscale\undefined%
  \makeatother%
  \begin{picture}(1,0.72361809)%
    \put(0,0){\includegraphics[width=\unitlength,page=1]{ex_CRBA.pdf}}%
    \put(0.62648486,0.44253923){\color[rgb]{0,0,0}\makebox(0,0)[lb]{\smash{$\{x,y\}$}}}%
    \put(0.69167746,0.41232035){\color[rgb]{0,0,0}\makebox(0,0)[lb]{\smash{}}}%
    \put(0.62458548,0.27349462){\color[rgb]{0,0,0}\makebox(0,0)[lb]{\smash{$\{x\}$}}}%
    \put(0.79030613,0.27301976){\color[rgb]{0,0,0}\makebox(0,0)[lb]{\smash{$\{y\}$}}}%
    \put(0.95602685,0.27301976){\color[rgb]{0,0,0}\makebox(0,0)[lb]{\smash{$\{z\}$}}}%
    \put(0.79039285,0.10685488){\color[rgb]{0,0,0}\makebox(0,0)[lb]{\smash{$\emptyset$}}}%
    \put(0.84094951,0.06220464){\color[rgb]{0,0,0}\makebox(0,0)[lb]{\smash{}}}%
    \put(0.79033849,0.44024649){\color[rgb]{0,0,0}\makebox(0,0)[lb]{\smash{$\{x,z\}$}}}%
    \put(0.95748388,0.43929682){\color[rgb]{0,0,0}\makebox(0,0)[lb]{\smash{$\{y,z\}$}}}%
    \put(0.78986378,0.60501751){\color[rgb]{0,0,0}\makebox(0,0)[lb]{\smash{$\{x,y,z\}$}}}%
  \end{picture}%
\endgroup%

%% file: ex_fk.pdf_tex
%% Creator: Inkscape inkscape 0.92.2, www.inkscape.org
%% PDF/EPS/PS + LaTeX output extension by Johan Engelen, 2010
%% Accompanies image file 'ex_fk.pdf' (pdf, eps, ps)
%%
%% To include the image in your LaTeX document, write
%%   \input{<filename>.pdf_tex}
%%  instead of
%%   \includegraphics{<filename>.pdf}
%% To scale the image, write
%%   \def\svgwidth{<desired width>}
%%   \input{<filename>.pdf_tex}
%%  instead of
%%   \includegraphics[width=<desired width>]{<filename>.pdf}
%%
%% Images with a different path to the parent latex file can
%% be accessed with the `import' package (which may need to be
%% installed) using
%%   \usepackage{import}
%% in the preamble, and then including the image with
%%   \import{<path to file>}{<filename>.pdf_tex}
%% Alternatively, one can specify
%%   \graphicspath{{<path to file>/}}
%% 
%% For more information, please see info/svg-inkscape on CTAN:
%%   http://tug.ctan.org/tex-archive/info/svg-inkscape
%%
\begingroup%
  \makeatletter%
  \providecommand\color[2][]{%
    \errmessage{(Inkscape) Color is used for the text in Inkscape, but the package 'color.sty' is not loaded}%
    \renewcommand\color[2][]{}%
  }%
  \providecommand\transparent[1]{%
    \errmessage{(Inkscape) Transparency is used (non-zero) for the text in Inkscape, but the package 'transparent.sty' is not loaded}%
    \renewcommand\transparent[1]{}%
  }%
  \providecommand\rotatebox[2]{#2}%
  \ifx\svgwidth\undefined%
    \setlength{\unitlength}{394.01574803bp}%
    \ifx\svgscale\undefined%
      \relax%
    \else%
      \setlength{\unitlength}{\unitlength * \real{\svgscale}}%
    \fi%
  \else%
    \setlength{\unitlength}{\svgwidth}%
  \fi%
  \global\let\svgwidth\undefined%
  \global\let\svgscale\undefined%
  \makeatother%
  \begin{picture}(1,1.03597122)%
    \put(0,0){\includegraphics[width=\unitlength,page=1]{ex_fk.pdf}}%
    \put(0.03235434,0.57621749){\color[rgb]{0,0,0}\makebox(0,0)[lb]{\smash{$1_R$}}}%
    \put(0.33127281,0.27972912){\color[rgb]{0,0,0}\makebox(0,0)[lb]{\smash{$Y_R$}}}%
    \put(0.01520373,0.28265189){\color[rgb]{0,0,0}\makebox(0,0)[lb]{\smash{$X_R$}}}%
    \put(0.30527636,0.00484571){\color[rgb]{0,0,0}\makebox(0,0)[lb]{\smash{$0_R$}}}%
    \put(0.23817145,0.62587217){\color[rgb]{0,0,0}\makebox(0,0)[lb]{\smash{$X_B {\lor} Y_B$}}}%
    \put(0.99024327,0.59030036){\color[rgb]{0,0,0}\makebox(0,0)[lb]{\smash{}}}%
    \put(0.2354522,0.33386576){\color[rgb]{0,0,0}\makebox(0,0)[lb]{\smash{$X_B$}}}%
    \put(0.52077709,0.33510875){\color[rgb]{0,0,0}\makebox(0,0)[lb]{\smash{$Y_B$}}}%
    \put(0.8099475,0.33895435){\color[rgb]{0,0,0}\makebox(0,0)[lb]{\smash{$Z_B$}}}%
    \put(0.52286266,0.04530236){\color[rgb]{0,0,0}\makebox(0,0)[lb]{\smash{$0_B$}}}%
    \put(1.2039493,0.08905556){\color[rgb]{0,0,0}\makebox(0,0)[lb]{\smash{}}}%
    \put(0.52666907,0.61886007){\color[rgb]{0,0,0}\makebox(0,0)[lb]{\smash{$X_B {\lor} Z_B$}}}%
    \put(0.8149711,0.6213074){\color[rgb]{0,0,0}\makebox(0,0)[lb]{\smash{$Y_B {\lor} Z_B$}}}%
    \put(0.52595086,0.91616901){\color[rgb]{0,0,0}\makebox(0,0)[lb]{\smash{$1_B$}}}%
    \put(1.30366,0.05960687){\color[rgb]{0,0,0}\makebox(0,0)[lb]{\smash{}}}%
    \put(1.61323119,0.54415306){\color[rgb]{0,0,0}\makebox(0,0)[lb]{\smash{}}}%
    \put(-0.21801854,0.15821438){\color[rgb]{0,0,0}\makebox(0,0)[lt]{\begin{minipage}{0.32702779\unitlength}\raggedright \end{minipage}}}%
  \end{picture}%
\endgroup%

%% file: ex_kripke.pdf_tex
%% Creator: Inkscape inkscape 0.92.2, www.inkscape.org
%% PDF/EPS/PS + LaTeX output extension by Johan Engelen, 2010
%% Accompanies image file 'ex_kripke.pdf' (pdf, eps, ps)
%%
%% To include the image in your LaTeX document, write
%%   \input{<filename>.pdf_tex}
%%  instead of
%%   \includegraphics{<filename>.pdf}
%% To scale the image, write
%%   \def\svgwidth{<desired width>}
%%   \input{<filename>.pdf_tex}
%%  instead of
%%   \includegraphics[width=<desired width>]{<filename>.pdf}
%%
%% Images with a different path to the parent latex file can
%% be accessed with the `import' package (which may need to be
%% installed) using
%%   \usepackage{import}
%% in the preamble, and then including the image with
%%   \import{<path to file>}{<filename>.pdf_tex}
%% Alternatively, one can specify
%%   \graphicspath{{<path to file>/}}
%% 
%% For more information, please see info/svg-inkscape on CTAN:
%%   http://tug.ctan.org/tex-archive/info/svg-inkscape
%%
\begingroup%
  \makeatletter%
  \providecommand\color[2][]{%
    \errmessage{(Inkscape) Color is used for the text in Inkscape, but the package 'color.sty' is not loaded}%
    \renewcommand\color[2][]{}%
  }%
  \providecommand\transparent[1]{%
    \errmessage{(Inkscape) Transparency is used (non-zero) for the text in Inkscape, but the package 'transparent.sty' is not loaded}%
    \renewcommand\transparent[1]{}%
  }%
  \providecommand\rotatebox[2]{#2}%
  \ifx\svgwidth\undefined%
    \setlength{\unitlength}{394.01574803bp}%
    \ifx\svgscale\undefined%
      \relax%
    \else%
      \setlength{\unitlength}{\unitlength * \real{\svgscale}}%
    \fi%
  \else%
    \setlength{\unitlength}{\svgwidth}%
  \fi%
  \global\let\svgwidth\undefined%
  \global\let\svgscale\undefined%
  \makeatother%
  \begin{picture}(1,1.03597122)%
    \put(0,0){\includegraphics[width=\unitlength,page=1]{ex_kripke.pdf}}%
    \put(0.13624759,0.52407015){\color[rgb]{0,0,0}\makebox(0,0)[lb]{\smash{$\omega_x$}}}%
    \put(0.48233498,0.52584157){\color[rgb]{0,0,0}\makebox(0,0)[lb]{\smash{$\omega_y$}}}%
    \put(0.82242456,0.52584157){\color[rgb]{0,0,0}\makebox(0,0)[lb]{\smash{$\omega_z$}}}%
    \put(0.78642619,0.5239636){\color[rgb]{0,0,0}\makebox(0,0)[lb]{\smash{}}}%
    \put(0.10650435,0.40178982){\color[rgb]{0,0,0}\makebox(0,0)[lb]{\smash{$p,q$}}}%
    \put(0.4429752,0.40021257){\color[rgb]{0,0,0}\makebox(0,0)[lb]{\smash{$\neg p, \neg q$ }}}%
    \put(0.82378469,0.4000604){\color[rgb]{0,0,0}\makebox(0,0)[lb]{\smash{$\neg p$}}}%
    \put(0,0){\includegraphics[width=\unitlength,page=2]{ex_kripke.pdf}}%
  \end{picture}%
\endgroup%

%% file: ex_RBA_2.pdf_tex
%% Creator: Inkscape inkscape 0.92.2, www.inkscape.org
%% PDF/EPS/PS + LaTeX output extension by Johan Engelen, 2010
%% Accompanies image file 'ex_RBA_2.pdf' (pdf, eps, ps)
%%
%% To include the image in your LaTeX document, write
%%   \input{<filename>.pdf_tex}
%%  instead of
%%   \includegraphics{<filename>.pdf}
%% To scale the image, write
%%   \def\svgwidth{<desired width>}
%%   \input{<filename>.pdf_tex}
%%  instead of
%%   \includegraphics[width=<desired width>]{<filename>.pdf}
%%
%% Images with a different path to the parent latex file can
%% be accessed with the `import' package (which may need to be
%% installed) using
%%   \usepackage{import}
%% in the preamble, and then including the image with
%%   \import{<path to file>}{<filename>.pdf_tex}
%% Alternatively, one can specify
%%   \graphicspath{{<path to file>/}}
%% 
%% For more information, please see info/svg-inkscape on CTAN:
%%   http://tug.ctan.org/tex-archive/info/svg-inkscape
%%
\begingroup%
  \makeatletter%
  \providecommand\color[2][]{%
    \errmessage{(Inkscape) Color is used for the text in Inkscape, but the package 'color.sty' is not loaded}%
    \renewcommand\color[2][]{}%
  }%
  \providecommand\transparent[1]{%
    \errmessage{(Inkscape) Transparency is used (non-zero) for the text in Inkscape, but the package 'transparent.sty' is not loaded}%
    \renewcommand\transparent[1]{}%
  }%
  \providecommand\rotatebox[2]{#2}%
  \ifx\svgwidth\undefined%
    \setlength{\unitlength}{481.88976378bp}%
    \ifx\svgscale\undefined%
      \relax%
    \else%
      \setlength{\unitlength}{\unitlength * \real{\svgscale}}%
    \fi%
  \else%
    \setlength{\unitlength}{\svgwidth}%
  \fi%
  \global\let\svgwidth\undefined%
  \global\let\svgscale\undefined%
  \makeatother%
  \begin{picture}(1,0.84705882)%
    \put(0,0){\includegraphics[width=\unitlength,page=1]{ex_RBA_2.pdf}}%
    \put(0.2318956,0.47736803){\color[rgb]{0,0,0}\makebox(0,0)[lb]{\smash{$1_R$}}}%
    \put(0.48668657,0.23651729){\color[rgb]{0,0,0}\makebox(0,0)[lb]{\smash{$\neg X_R$}}}%
    \put(0.08402441,0.23733498){\color[rgb]{0,0,0}\makebox(0,0)[lb]{\smash{$X_R$}}}%
    \put(0.36166721,0.01018749){\color[rgb]{0,0,0}\makebox(0,0)[lb]{\smash{$0_R$}}}%
    \put(0.4001814,0.51796803){\color[rgb]{0,0,0}\makebox(0,0)[lb]{\smash{$Y_B$}}}%
    \put(0.80966949,0.48265735){\color[rgb]{0,0,0}\makebox(0,0)[lb]{\smash{}}}%
    \put(0.54580154,0.04326675){\color[rgb]{0,0,0}\makebox(0,0)[lb]{\smash{$0_B$}}}%
    \put(0.98440561,0.07281602){\color[rgb]{0,0,0}\makebox(0,0)[lb]{\smash{}}}%
    \put(0.7566286,0.51423566){\color[rgb]{0,0,0}\makebox(0,0)[lb]{\smash{$\neg Y_B$}}}%
    \put(0.57322845,0.75532839){\color[rgb]{0,0,0}\makebox(0,0)[lb]{\smash{$1_B$}}}%
    \put(1.06593377,0.04873738){\color[rgb]{0,0,0}\makebox(0,0)[lb]{\smash{}}}%
    \put(1.31905374,0.44492515){\color[rgb]{0,0,0}\makebox(0,0)[lb]{\smash{}}}%
    \put(-0.17826222,0.12936352){\color[rgb]{0,0,0}\makebox(0,0)[lt]{\begin{minipage}{0.26739331\unitlength}\raggedright \end{minipage}}}%
    \put(0.55654948,0.24683129){\color[rgb]{0,0,0}\makebox(0,0)[lb]{\smash{}}}%
  \end{picture}%
\endgroup%